%% file: amgs.tex
\documentclass[a4paper]{article}
\input{macros.tex}
\usepackage[affil-it]{authblk}
\usepackage{framed}
\usepackage[numbers]{natbib}

\begin{document}
\title{Abstract machines for game semantics, revisited\footnote{An extended abstract
of this paper is due to appear in the
Twenty-Eighth Annual ACM/IEEE Symposium on Logic in Computer Science (LICS 2013), June 25-28, 2013, New Orleans, USA.}}

\author{Olle Fredriksson}
\author{Dan R. Ghica}
\affil{University of Birmingham, UK}
\maketitle
\begin{abstract}
We define new abstract machines for game semantics which
correspond to networks of conventional computers, and can be used as an
intermediate representation for compilation targeting distributed systems. This is
achieved in two steps. First we introduce the HRAM, a \emph{Heap and Register
Abstract Machine}, an abstraction of a conventional computer, which can be
structured into HRAM nets, an abstract point-to-point network model. HRAMs are
multi-threaded and subsume communication by tokens (\emph{cf.} IAM) or jumps.
Game Abstract Machines (GAM), are HRAMs with additional
structure at the interface level, but no special operational capabilities. We
show that GAMs cannot be naively composed, but composition must be mediated
using appropriate HRAM combinators. HRAMs are flexible enough to allow the
representation of game models for languages with state (non-innocent games) or
concurrency (non-alternating games). We illustrate the potential of this
technique by implementing a toy distributed compiler for ICA, a higher-order
programming language with shared state concurrency, thus significantly
extending our previous distributed PCF compiler. We show that compilation is
sound and memory-safe, i.e. no (distributed or local) garbage collection is
necessary.
\end{abstract}

\section{Introduction}

One of the most profound discoveries in theoretical computer science is the
fact that logical and computational phenomena can be subsumed by relatively
simple communication protocols. This understanding came independently from
Girard's work on the Geometry of Interaction (GOI)~\cite{girard1989geometry}
and Milner's work on process calculi~\cite{DBLP:conf/icalp/Milner90}, and
had a profound influence on the subsequent development of game semantics
(see~\cite{DBLP:conf/lics/Ghica09} for a historical survey). Of the three, game
semantics proved to be particularly effective at producing precise mathematical
models for a large variety of programming languages, solving a long-standing
open problem concerning higher-order sequential
computation~\cite{DBLP:journals/iandc/AbramskyJM00,DBLP:journals/iandc/HylandO00}.

One of the most appealing features of game semantics is that it has a dual
denotational and operational character. By \emph{denotational} we mean that it
is compositionally defined on the syntax and by \emph{operational} we mean that
it can be effectively presented and can form a basis for
compilation~\cite{DBLP:conf/memocode/Ghica11}. This feature was apparent from
the earliest presentations of game semantics~\cite{DBLP:conf/fpca/HylandO95}
and is not very surprising, although the operational aspects are less
perspicuous than in interpretations based on process calculi or GOI, which
quickly found applications in compiler~\cite{DBLP:conf/popl/Mackie95} or
interpreter~\cite{DBLP:journals/jfp/Benton05} development and optimisation. 

An important development, which provided essential inspiration for this work,
was the introduction of the \emph{Pointer Abstract Machine} (PAM) and the
\emph{Interaction Abstract Machine} (IAM), which sought to fully restore the
operational intuitions of game semantics~\cite{DBLP:conf/lics/DanosHR96} by
relating them to two kinds of abstract machines, one based on term rewriting
(PAM) and one based on networks of automata (IAM) profoundly inspired by GOI. A
further optimisation of IAM, the \emph{Jumping Abstract Machine} (JAM) was
introduced subsequently to avoid the overheads of the
IAM~\cite{DBLP:journals/tcs/DanosR99}.

\paragraph*{Contribution} In this paper we are developing the line of work on
the PAM, IAM, and JAM, in order to define new abstract machines which correspond more
closely to \emph{networks of conventional computers} and can be used as
an intermediate representation for compilation targeting distributed systems. This is
achieved in two steps. First we introduce the HRAM, a \emph{Heap and Register
Abstract Machine}, an abstraction of a conventional computer, which can be
structured into HRAM nets, an abstract point-to-point network model. HRAMs are
multi-threaded and subsume communication by tokens (cf. IAM) or jumps.
GAMs, \emph{Game Abstract Machines}, are
HRAMs with additional structure at the interface level, but no special
operational capabilities. We show that GAMs cannot be naively composed, but
composition must be mediated using appropriate HRAM combinators. Starting from a
formulation of game semantics in the nominal
model~\cite{DBLP:journals/entcs/GabbayG12} has two benefits. First, pointer
manipulation requires no encoding or decoding, as in integer-based
representations, but exploits the HRAM ability to create locally fresh
\emph{names}. Second, token size is constant as only names are passed around;
the computational history of a token is stored by the HRAM rather than passing
it around (cf. IAM). HRAMs are also flexible enough to allow the representation
of game models for languages with state (\emph{non-innocent} games) or
concurrency (\emph{non-alternating} games). We illustrate the potential of this
technique by implementing a compiler targeting distributed systems for ICA, a higher-order
programming language with shared state
concurrency~\cite{DBLP:journals/apal/GhicaM08}, thus significantly extending
our previous distributed PCF compiler~\cite{GhicaF12}. We show that compilation
is sound and memory-safe, i.e.\ no (distributed or local) garbage collection is
necessary.\footnote{Available from \url{http://veritygos.org/gams}.}

\paragraph*{Other related and relevant work} The operational intuitions of GOI
were originally confined to the sequential setting, but more recent work on
Ludics showed how they can be applied to
concurrency~\cite{DBLP:conf/lics/FaggianM05} through an abstract treatment not
immediately applicable to our needs. Whereas our work
takes the IAM/JAM as the starting point, developing abstract machines akin to
the PAM revealed interesting syntactic and operational connections between game
semantics and B\"{o}hm trees~\cite{DBLP:journals/corr/abs-0706-2544}. The
connection between game semantics, syntactic recursion schemes and automata
also had several interesting applications to verifying higher-order computation
(see e.g.~\cite{DBLP:conf/esop/Ong08}). Finally, the connection between game
semantics and operational semantics can be made more directly by eliding all
the semantic structure in the game and reducing them to a very simple
communication mechanism between a program and its environment, which is useful
in understanding hostile opponents and verifying security
properties~\cite{DBLP:journals/entcs/GhicaT12}.

\section{Simple nets}
In this section we introduce a class of basic abstract machines for
manipulating heap structures, which also have  primitives for communications
and control. They represent a natural intermediate stage for
compilation to machine language, and will be used as such in
Sec.~\ref{sec:disc}. The machines can naturally be organised into
communication networks which give an abstract representation of distributed
systems. We find it formally convenient to work in a nominal model in order to
avoid the difficulties caused by concrete encoding of game structures,
especially \emph{justification pointers}, as integers. We assume a certain
familiarity from the reader with basic nominal concepts. The interested reader
is referred to the literature (\cite{DBLP:conf/lics/GabbayP99} is a starting
point).

\subsection{Heap and register abstract machines (HRAM)}
We fix a set of \emph{port names} ($\portnames$) and a set of \emph{pointer
names} ($\pointers$) as disjoint sets of atoms. Let $L \defeq \{\lO, \lP \}$ be
the set of polarities of a port. To maintain an analogy with game semantics from
the beginning, port names correspond to game-semantic \emph{moves} and
input/output polarities correspond to opponent/proponent. A \emph{port
structure} is a tuple $\port l a \in \ports = L \tprod \portnames$.
An \emph{interface} $A \in \psetfin(\ports)$ is a set of port structures such that
all port names are unique, i.e. $\forall p = \port la, p'=\port {l'}{a'}
\in A$, if $a=a'$ then $p=p'$. Let the support of an interface be $\supp(A)
\defeq \{a \mid \port la \in A\}$, its set of port names.

The \emph{tensor} of two interfaces is defined as  $A \tensor B \defeq A \cup
B$, where $\supp(A) \cap \supp(B) = \emptyset$. The dual of an interface is
defined as $A \dual \defeq \{ p \dual \mid p \in A\}$ where $\port la \dual
\defeq \port {l\dual}a$, $\lO \dual \defeq \lP$ and $\lP \dual \defeq \lO$.
An arrow interface is defined in terms of tensor and dual, $A \fun B \defeq A
\dual \tensor B$.

We introduce notation for opponent ports of an interface  $A \at O \defeq
\{\port \lO a \in A\}$. The player ports of an interface $A \at P$ is defined
analogously. The set of all interfaces is denoted by~$\interfaces$. We say that
two interfaces \emph{have the same shape} if they are equivariant, i.e.\ there
is a permutation $\pi:\portnames\rightarrow\portnames$ such that $\{\pi \cdot
p  \mid p \in A_1 \} = A_2$, and we write $\pi\vdash A_1 \pneq A_2$, where $\pi
\cdot \port la \defeq \port l{\pi(a)}$ is the permutation action of $\pi$. We
may only write $A_1 \pneq A_2$ if $\pi$ is obvious or unimportant. 

Let the set of data $\datas$ be $\emptyset \in \mathbbm 1$, pointer names $a
\in \pointers$ or integers $n\in\integers$. Let the set of instructions
$\instrs$ be as below, where $i,j,k\in\mathbb N + \mathbbm 1$ (which
permits ignoring results and allocating ``null'' data).
\begin{itemize}
  \item
    $i \leftarrow \nialloc\ j,k$ allocates a new pointer in the heap and
    populates it with the values stored in registers $j$ and $k$, storing the
    pointer in register $i$.
  \item
    $i,j \leftarrow \niget\ k$ reads the tuple pointed at by the name in the
    register $k$ and stores it in registers $i$ and $j$.
  \item
    $\niupdate\ i,j$ writes the value stored in register $j$ to the
    second component of the value pointed to by the name in register $i$.
  \item
    $\nifree\ i$ releases the memory pointed to by the name in the register $i$
    and resets the register.
  \item
    $\niflip\ i,j$ flips the values of registers $i$ and $j$.
  \item
    $i \leftarrow \iset\ j$ sets register $i$ to value $j$.
\end{itemize}

Let code fragments $\codes$ be $\codes ::= \instrs \seq \codes \mid \niifzero\
{\mathbb N}\ \codes\ \codes \mid \ileap a \mid \iend$. The port names occurring
in the code fragment are $\supp \in \codes \tfun \psetfin(\portnames)$, defined
in the obvious way (only the $\nispark\ a$ instruction can contribute names).
An $\niifzero\ i$ instruction will branch according to the value stored in
register~$i$. A $\nispark\ a$ will either jump to $a$ or send a message to $a$,
depending on whether $a$ is a local port or not.

An \emph{engine} is an interface together with a port map, $E=\engine AP \in
\interfaces \tprod (\supp(A \at O) \tfun \codes)$ such that for each code
fragment $c \in \cod\ P$ and each port name $a \in \supp(c)$, $\port \lP a \in
A$, meaning that ports that are ``sparked'' must be output ports of the
interface $A$. The set of all engines is $\engines$.

Engines have threads and shared heap. All threads have a fixed number of
registers $r$, which is a global constant. For the language ICA we will need
four registers, but languages with more kinds of pointers in the game model,
e.g.\ control pointers~\cite{DBLP:conf/esop/Laird02}, may need and use more
registers.

A \emph{thread} is a tuple $t=(c,\overline d)\in T = \codes \tprod
\datas^r$: a code fragment and an $r$-tuple of data register values.

An \emph{engine configuration} is a tuple $k = \engineconf {\many t}h \in
\engineconfs = \psetfin(T) \tprod (\pointers \partialfun \pointers \tprod \datas)$:
a set of threads and a heap that maps pointer names to pairs of pointer names
and data items.

A pair consisting of an engine configuration and an engine will be written
using the notation $k \hastype E \in \engineconfs \tprod \engines$.  Define the
function $\initial \in \engines \tfun \engineconfs \tprod \engines$ as
$\initial(E) \defeq \engineconf \emptyset \emptyset \hastype E$ for an engine
$E$. This function pairs the engine up with an engine configuration consisting
of no threads and an empty heap.

HRAMs communicate using \emph{messages}, each consisting of a port name
and a vector of data items of size $r_m$: $m = (x,\overline d)\in \msgs =
\portnames \tprod \datas^{r_m}$.  The constant $r_m$ specifies the size of the
messages in the network, and has to fulfil $r_m \leq r$. For a set $X \subseteq
\portnames$, define $\msgs_X = X \tprod  \datas^{r_m}$, the subset of $\msgs$
whose port names are limited to those of $X$.

We specify the operational semantics of an engine $E = \engine AP$ as a
transition relation $- \enginestep E \conn - - \subseteq \engineconfs \tprod
(\{\bullet\} \cup (L \tprod \msgs)) \tprod \engineconfs$.  The relation is
either labelled with $\bullet$ --- a silent transition --- or a polarised
message --- an observable transition. The messages will be constructed simply
from the first $r_m$ registers of a thread, meaning that on certain actions
part of the register contents become observable in the transition relation.

To aid readability, we use the following shorthands:
\begin{itemize}
  \item
    $n \enginestep E \conn {} {n'}$ means $n \enginestep E \conn {\bullet}
    {n'}$ (silent transitions).
  \item
    $n \enginestep E \conn {\rmsg a{\many d}} {n'}$ means $n \enginestep E
    \conn {\tuple {\lP, \rmsg a{\many d}}} {n'}$ (output transitions).
  \item
    $n \enginestep E \conn {{\rmsg a{\many d}}^\bullet} {n'}$ means $n
    \enginestep E \conn {\tuple {\lO, \rmsg a{\many d}}} {n'}$ (input
    transitions).
\end{itemize}

We use the notation $\many d$ for $n$-tuples of registers and then $d_i$
for the (zero-based) $i$-th component of $\many d$, and $d_\emptyset \defeq \emptyset$. For
updating a register, we use $\many d \regupd i d \defeq (d_0, \cdots, d_{i-1},
d, d_{i+1}, \cdots, d_{n-1})$ and $\many d \regupd \emptyset d \defeq \many d$.

To construct messages from the register contents of a thread, we use the
functions $\msgf \in \datas^{r} \tfun \datas^{r_m}$, which takes the
first $r_m$ components of its input, and $\regsf \in \datas^{r_m} \tfun
\datas^{r}$, which pads its input with $\emptyset$ at the end (i.e.
$\regsf(\many d) \defeq \tuple{d_0,\ldots,d_{r_m-1},\emptyset,\ldots}$).

The network connectivity is specified by the function $\conn$, which will be
described in more detail in the next sub-section. For a port name $a$,
$\conn(a)$ can be read as ``the port that $a$ is connected to''.
The full operational rules for HRAMs are given in Fig.~\ref{fig:hramop}.
The interesting rule is that for $\nispark$ because it depends on whether the
port where the next computation is ``sparked'' is local or not. If the port is
local then $\nispark$ makes a jump, and if the port is non-local then it
produces an output token and the current thread of execution is terminated,
similar to the IAM.
\input{hramop.tex}

\begin{figure}
\centering

  \begin{tikzpicture}
    \node[scale=0.9]{
      \input{hramnet.latex}
    };
  \end{tikzpicture}

\caption{Example HRAM net}
\label{fig:exhramnet}
\end{figure}
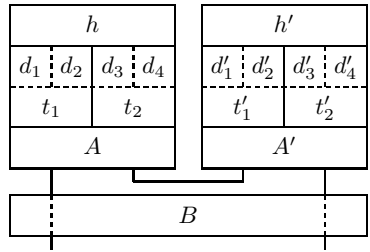

\subsection{HRAM nets}

A well-formed \emph{HRAM net} $S\in \nets$ is a set of engines, a function
over port names specifying what ports are connected, and an external interface,
$ S = \net {\many E}\conn A$, where $E \in \engines, A\in \interfaces$, and
$\conn$ is a bijection between the net's output and input port names. Specifically, $\conn$ has to
be in $ \supp(A \at O \tensor A_{\many E}\at P) \tfun
        \supp(A \at P \tensor A_{\many E}\at O)$,
where $A_{\many E} = \tensor\{A \mid \engine A P \in \many E\}$.

Fig.~\ref{fig:exhramnet} shows a diagram of an HRAM net with two HRAMs
(interfaces $A, A'$, two ports each), each with two running threads ($t_i,
t'_i$) with local registers ($d_i, d'_i$) and shared heaps ($h, h'$). Two of
the HRAM ports are connected and two are part of the global interface $B$.

The function $\conn$ gives the net connectivity.  It being in $\supp(A \at O
\tensor A_{\many E}\at P) \tfun \supp(A \at P \tensor A_{\many E}\at O)$ means
that it maps each input port name of the net's interface and output port name
of the net's engines to either an output port name of the net's interface or an
input port name of one of its engines.  Since it is a bijection, each port name
(and thus port) is connected to exactly one other port name, so the abstract
network model we are using is point-to-point.

For an engine $e = \engine A P$, we define a \emph{singleton} net with $e$ as
its sole engine as $ \singleton(e) = \net {\{e\}}\conn {A'}$,
where $A'$ is an interface such that $\conn \vdash A \pneq A'$ and $\conn$ is
given by:
\begin{align*}
    \conn(a) & \defeq \pi(a)      \text{ if } a \in \supp(A\at P)
 \\ \conn(a) & \defeq \pi^{-1}(a) \text{ if } a \in \supp({A'}\at O)
\end{align*}

A \emph{net configuration} is a set of tuples of engine configurations and engines
and a multiset of pending messages: $n=\netconf{\many {e \hastype E}}{\many m} \in
\netconfs = \psetfin(\engineconfs \tprod \engines) \tprod \msetfin(\msgs)$.
Define the function $\initial \in \nets \tfun \netconfs$ as $\initial\net
{\many E}\conn A \defeq \netconf {\{\initial(E) \mid E \in \many E\}} \emptyset$,
a net configuration with only initial engines and no pending messages.

The operational semantics of a net $S=\net {\many E} \conn A$ is specified as a
transition relation $- \netstep - - \subseteq \netconfs \tprod (\{\bullet\}
\cup (L \tprod \msgs_{\supp(A)})) \tprod \netconfs$. The semantics is given in
the style of the Chemical Abstract Machine
(CHAM)~\cite{DBLP:conf/popl/BerryB90}, where HRAMs are ``molecules'' and the
pending messages of the HRAM net is the ``solution''. HRAM inputs (outputs) are
to (from) the set of pending messages. Silent transitions of any HRAM are
silent transitions of the net. The rules are given in Fig.~\ref{fig:netop}.

\input{netop.tex}

\subsection{Semantics of HRAM nets}
We define $\List[A]$ for a set $A$ to be finite sequences of elements from $A$,
and use $s \snoc s'$ for concatenation.  A \emph{trace} for a net $\net {\many
E}\conn A$ is a \emph{finite sequence} of messages with polarity: $s \in
\List[L \tprod \msgs_{\supp(A)}]$. Write $\alpha \in L \tprod \msgs_{\supp(A)}$
for single polarised messages. We use the same notational convention as before
to identify inputs ($-^\bullet$). 

For a trace $s = \alpha_1 \snoc \alpha_2\snoc  \cdots\snoc  \alpha_n $, define $\netstep{s}$
to be the following composition of relations on net configurations:
$
  \netstep{\alpha_1}\netstep{}^*
  \netstep{\alpha_2}\netstep{}^*
  \cdots \netstep{\alpha_n}
$,
where $\netstep{}^*$ is the reflexive transitive closure of
$\netstep{}$, i.e.\ any number of silent steps are
allowed in between those that are observable.

Write $\traces_A$ for the set $\List[L \tprod \msgs_{\supp(A)}]$.  The
denotation $\denot S \subseteq \traces_A$ of a net $S = \net {\many E}\conn A$
is the set of traces of observable transitions reachable
from the initial net configuration $\initial(S)$ using the transition relation:
\[ \denot S \defeq \{s \in \traces_A \mid \exists n. \initial(S) \netstep{s}
n\} \] 
The denotation of a net includes the empty trace and is
prefix-closed by construction.

As with interfaces, we are not interested in the actual port names occurring
in a trace, so we define \emph{equivariance} for sets of traces. Let $S_1
\subseteq \traces_{A_1}$ and $S_2 \subseteq \traces_{A_2}$ for $A_1,A_2 \in
\interfaces$. $S_1 \pneq S_2$ if and only if there is a permutation $\pi \in
\portnames \tfun
\portnames$ such that $\{\pi \cdot s \mid s \in S_1\} = S_2$, where
$\pi \cdot \emptytrace \defeq \emptytrace$ and
$\pi \cdot (s \snoc \tuple {l, \rmsg a{\many d}}) \defeq (\pi \cdot s) \snoc \tuple {l, \rmsg
{\pi(x)}{\many d}}$.

Define the \emph{deletion} operation $s \ndel A$ which removes from a trace all
elements $\tuple{l, \rmsg x{\many d}}$ if $ x \in \supp(A)$ and define the
interleaving of sets of traces $S_1 \subseteq \traces_A$ and $S_2
\subseteq \traces_B$ as
$
  S_1 \tensor S_2 \defeq
    \{ s \mid s \in \traces_{A\tensor B}
      \wedge s \ndel B \in S_1
      \wedge s \ndel A \in S_2 \}
$. 

Define the composition of the sets of traces $S_1 \subseteq \traces_{A \fun B}$
and $S_2 \subseteq \traces_{B' \fun C}$ with $\pi \vdash B \pneq B'$ as the usual \emph{synchronisation and hiding} in trace semantics:
\[
  S_1 \compose S_2 \defeq
    \{s \ndel B \mid s \in \traces_{A \tensor B \tensor C}
      \wedge s \ndel C \in S_1
      \wedge \pi \cdot {s^{*B}} \ndel A \in S_2 \}
\]
(where $s^{*B}$ is $s$ where the messages from $B$ have reversed polarity.)

Two nets, $f = \net {\many E_f}{\conn_f}{I_f}$ and $g = \net {\many E_g}{\conn_g}{I_g}$
are said to be \emph{structurally equivalent} if they are graph-isomorphic, i.e.\
$\pi \cdot \many E_f = \many E_g$, $\pi \vdash I_f \pneq I_g$
and $\conn_g\circ\pi=\pi\circ\conn_f$. 

\begin{theorem} \label{thm:structisdenot}
If $S_1$ and $S_2$ are structurally equivalent nets, then $\denot{S_1} \pneq
\denot{S_2}$.
\end{theorem}
\begin{proof}
  A straightforward induction on the trace length, in both directions.
\end{proof}

\subsection{HRAM nets as a category}

In this sub-section we will show that HRAM nets form a symmetric compact-closed
category.  This establishes that our definitions are sensible and that HRAM
nets are equal up to topological isomorphisms. This result also shows that the
structure of HRAM nets is very loose.

\input{category.tex}

The following result explicates how communicating HRAMs can be combined into
a single machine, where the intercommunication is done with jumping rather
than message passing, in a sound way:

\begin{theorem} \label{thm:combination}
  If $E_1 = \engine {A_1} {P_1}$ and $E_2 = \engine {A_2} {P_2}$ are
  engines and $S = \net {\{E_1,E_2\}} \conn A$ is a net, then
  $E_{12} = \engine {A_1 \tensor A_2} {P_1 \cup P_2}$ is an engine,
  $S' = \net {\{E_{12}\}} \conn A$ is a net and
  $\denot {S} \subseteq \denot{S'}$.
\end{theorem}
\input{combination.tex}

We define a family of projection HRAM nets $\Pi_{i,{A_1 \tensor \cdots
\tensor A_n}} \hastype A_1 \tensor \cdots \tensor A_n \tfun A_i$ by first
constructing a family of ``sinks'' $\sink_A \hastype A \tfun I \defeq
\singleton(\engine {A \fun I} P)$ where $I = \emptyset$ and $P(a) = \iend$
for each $a$ in its domain and then defining e.g. $\Pi_{1,A \tensor B}
\hastype A \tensor B \tfun A \defeq \id A \tensor \sink_B$.

\section{Game nets for ICA}\label{sec:gamn}
The structure of a \netcat token is determined by the number of registers $r$
and the message size $r_m$, which are globally fixed. To implement
game-semantic machines we require four message components: a port name, two pointer
names, and a data fragment, meaning that $r_m = 3$.
We choose $r = 4$, to get an additional register for temporary thread values to
work with. From this point on, messages in nets and traces will be restricted
to this form.

The message structure is intended to capture the structure of a
move when game semantics is expressed in the nominal model. The port name is
the move, the first name is the ``point'' whereas the second name is the
``butt'' of a justification arrow, and the data is the value of the move. This
direct and abstract encoding of the justification pointer as names is quite
different to that used in PAM and in other GOI-based token machines. In PAM the
pointer is represented by a sequence of integers encoding the hereditary
justification of the move, which is a snap-shot of the computational causal
history of the move, just like in GOI-based machines. Such encodings have an
immediate negative consequence, as tokens can become impractically large in
complex computations, especially involving recursion. Large tokens entail
not only significant communication overheads but also the computational overheads of
decoding their structure. A subtler negative consequence of such an encoding is
that it makes supporting the semantic structures required to interpret state
and concurrency needlessly complicated and inefficient. The nominal
representation is simple and compact, and efficiently exploits local machine
memory (heap) in a way that previous abstract machines, of a ``functional''
nature, do not.

The price that we pay is a failure of compositionality, which we will illustrate
shortly. The rest of the section will show how compositionality can be restored
without substantially changing the HRAM framework. If in HRAM nets
compositionality is ``plug-and-play'', as apparent from its compact-closed
structure, \emph{Game Abstract Machine} (GAM) composition must be mediated by a
family of operators which are themselves HRAMs.

In this simple motivating example it is assumed that the reader is familiar
with game semantics, and several of the notions to be introduced formally in
the next sub-sections are anticipated. We trust that this will be not
confusing.

Let $S$ be a HRAM representing the game semantic model for the \emph{successor}
operation $S:int\rightarrow int$. The HRAM net in Fig.~\ref{fig:nlhc}
represents a (failed) attempt to construct an interpretation for the term
$x:int\vdash S(S(x)):int$ in a context $C[-_{int}]:int$. This is the standard
way of composing GOI-like machines.

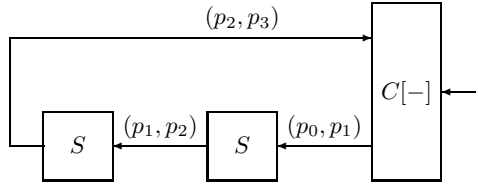
\begin{figure}
\centering

  \begin{tikzpicture}
    \node[scale=0.9]{
      \input{nonlocal.latex}
    };
  \end{tikzpicture}

\caption{Non-locality of names in HRAM composition}
\label{fig:nlhc}
\end{figure}

The labels along the edges of the HRAM net trace a token $(a, p_0, p_1, d)$
sent by the context $C[-]$ in order to evaluate the term. We elide $a$ and $d$,
which are irrelevant, to keep the diagram uncluttered. The token is received by
$S$ and propagated to the other $S$ HRAM, this time with tokens $(p_1, p_2)$.
This trace of events $(p_0,p_1)\snoc(p_1,p_2)$ corresponds to the existence of
a justification pointer from the second action to the first in the game model.
The essential correctness invariant for a well-formed trace representing a
game-semantic play is that each token consists of a \emph{known} name and a
\emph{fresh} name (if locally created, or \emph{unknown} if externally
created). However, the second $S$ machine will respond with $(p_2,p_3)$ to
$(p_1,p_2)$, leading to a situation where $C[-]$ receives a token formed from
two unknown tokens.

In game semantics, the composition of $(p_0,p_1)\snoc(p_1,p_2)$ with
$(p_1,p_2)\snoc(p_2,p_3)$ should lead to $(p_0,p_1)\snoc(p_1,p_3)$, as
justification pointers are ``extended'' so that they never point into a move
hidden through composition. This is precisely what the composition operator, a
specialised HRAM, will be designed to achieve.

\subsection{Game abstract machines (GAM) and nets}

\begin{definition}\label{def:gameint}
We define a \emph{game interface} (cf. \emph{arena}) as a tuple
$
\mf A = \ginterface {A}{\qst A}{\ini A}{\enables A}$
where 
\begin{itemize}
  \item $A \in \interfaces$ is an interface. For game interfaces $\mf A, \mf B, \mf C$
  we will write $A,B,C$ and so on for their underlying interfaces.
  \item The set of ports is partitioned into a subset of question port names $\qst A$ 
and one of answer port names $\ans A$, $\qst A\uplus\ans A=\supp(A)$.
  \item The set of initial port names $\ini A$ is a subset of the $\lO$-labelled
        question ports.
  \item The enabling relation $\enables A$ relates question port names to
        non-initial port names such that if $a \enables A a'$ for port names $a
        \in \qst A$ with $\port la \in A$ and $a' \in \supp(A) \setminus \ini A$
        with $\port {l'}{a'}\in A$, then $l \neq l'$.
\end{itemize}
\end{definition}
For notational consistency, write $\opp A \defeq \supp(A \at O)$ and
$\prop A \defeq \supp(A \at P)$.
Call the set of all game interfaces $\ginterfaces$.
Game interfaces are equivariant, $ \pi \vdash \mf A \pneq \mf B $,
  if and only if
$\pi \vdash A  \pneq B$, $ \{\pi(a) \mid a \in \qst A\} = \qst B$,
$\{\pi(a) \mid a \in \ini A\}  = \ini B$ and $\{\tuple{\pi(a),\pi(a')}\mid a \enables A a'\}  = \app \enables B $.
\begin{definition}\label{def:gamcon}
  For game interfaces (with disjoint sets of port names) $\mf A$ and $\mf B$, we define:
  \begin{align*}
       \mf A \tensor \mf B & \defeq
       \ginterface {A \tensor B} {\qst A {\cup} \qst B}
                   {\ini A {\cup} \ini B}{\enables A {\cup} \enables B}
    \\ \mf A \fun \mf B & \defeq
       \ginterface {A \fun B} {\qst A {\cup} \qst B}
                   {\ini B} {\enables A {\cup} \enables B {\cup}
                             (\ini B \times \ini A)}.
  \end{align*}
\end{definition}
A \emph{GAM net} is a tuple $G = \gnet S {\mf A} \in \nets \tprod \ginterfaces$
consisting of a net and a game interface such that $S = \net {\many E} \conn
A$, i.e.  the interface of the game net is the same as that of the game
interface.
The denotational semantics of a GAM net $G = \gnet S {\mf A}$
is just that of the underlying HRAM net:
$
\denot{G} \defeq \denot{S}
$.

\subsection{Game traces}
To be able to use game semantics as the specification for game nets we define
the usual legality conditions on traces, following~\cite{DBLP:journals/entcs/GabbayG12}.

\input{legality.tex}
\subsection{Copycat}
The quintessential game-semantic behaviour is that of the \emph{copy-cat strategy}, as it appears in various guises in the representation of all structural morphisms of any category of strategies. A copy-cat not only replicates the behaviour of its Opponent in terms of moves, but also in terms of justification structures. Because of this, the copy-cat strategy needs to be either history-sensitive (stateful) or the justification information needs to be carried along with the token. We take the former approach, in contrast to IAM and other GOI-inspired machines.

Consider the identity (or copycat) strategy on $\com \fun \com$, where $\com $ is a two-move arena (one question, one answer).
A typical play
may look as in Fig.~\ref{fig:tipp}.
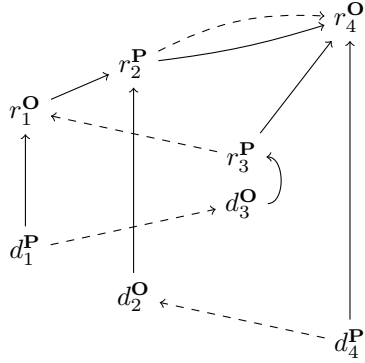
\begin{figure}
\begin{center}
  \begin{tikzpicture}[descr/.style={text=black}]
    \matrix [column sep=0.5em]
    {
         \node(c1){$(\com_1$}; & \node(c2){$\com_2)$}; & \node(c3){$(\com_3$}; & \node(c4){$\com_4)$};
      \\                       &                       &                       & \node(r4){$r_4^\lO$};
      \\                       & \node(r2){$r_2^\lP$}; &                       &
      \\ \node(r1){$r_1^\lO$}; &                       &                       &
      \\                       &                       & \node(r3){$r_3^\lP$}; &
      \\                       &                       & \node(d3){$d_3^\lO$}; &
      \\ \node(d1){$d_1^\lP$}; &                       &                       &
      \\                       & \node(d2){$d_2^\lO$}; &                       &
      \\                       &                       &                       & \node(d4){$d_4^\lP$};
    \\ };
    \path[color=white]
      (c1) edge node[descr]{$\Rightarrow$} (c2)
      (c2) edge node[descr]{$\rightarrow$} (c3)
      (c3) edge node[descr]{$\Rightarrow$} (c4)
      ;
    \path[->] 
      (r2) edge [bend right=5] (r4)
      (r1) edge (r2)
      (r3) edge (r4)
      (d3) edge [bend right=90] (r3)
      (d1) edge (r1)
      (d2) edge (r2)
      (d4) edge (r4)
    ;
    \path[dashed,->] 
      (r2) edge [bend left=15] (r4)
      (r3) edge (r1)
      (d1) edge (d3)
      (d4) edge (d2)
    ;
  \end{tikzpicture}
\end{center}
\vspace{-.5cm}
\caption{A typical play for copycat}
\label{fig:tipp}
\end{figure}
The full lines represent justification pointers, and the trace (play) is represented nominally as
\[
(r_4,p_0,p_1)\snoc
(r_2,p_1,p_2)\snoc
(r_1,p_2,p_3)\snoc
(r_3,p_1,p_4)\snoc
(d_3,p_4)\cdots
\]
To preserve the justification structure, a copycat engine only needs to store ``copycat links'', which are shown
as dashed lines in the diagram between question moves. 
In this instance,
for an input on $r_4$, a heap value mapping a freshly created $p_2$
(the pointer to $r_2$) to $p_1$ (the pointer from $r_4$) is added.

The reason for mapping $p_2$ to $p_1$ becomes clear when the engine later gets
an input on $r_1$ with pointers $p_2$ and $p_3$. It can then replicate the move
to $r_3$, but using $p_1$ as a justifier. By following the $p_2$ pointer
in the heap it gets $p_1$ so it can produce $(r_3,p_1,p_4)$, where $p_4$ is a fresh heap value mapping to $p_3$. When receiving
an answer, i.e.  a $d$ move, the copycat link can be dereferenced and then
\emph{discarded} from the heap.

The following HRAM macro-instructions are useful in defining copy-cat machines to, respectively, handle the pointers in an initial question, a non-initial question and an answer:
\begin{align*}
     \icopyi & \defeq \niflip\ 0,1 \seq 1 \leftarrow \nialloc\ 0,3
  \\ \icopyq & \defeq 1 \leftarrow \nialloc\ 1,3 \seq 0,3 \leftarrow \niget\ 0
  \\ \icopya & \defeq \niflip\ 0,1 \seq 0,3 \leftarrow \niget 1 \seq \nifree\ 1
\end{align*}
For game interfaces $\mf A$ and $\mf{A'}$ such that $\pi \vdash \mf A \pneq \mf{A'}$, we define a
generalised copycat engine as $ \copycat_{C,\pi,\mf A} = \engine {A \fun A'} P$, where:
\begin{align*}
  P & \defeq
  \{ q_2 \mapsto C \seq \ileap {q_1}
      \mid q_2 \in \ini {A'} \wedge q_1 = \pi^{-1} (q_2) \}
  \\ & \cup
  \{ q_2 \mapsto \icopyq \seq \ileap {q_1}
      \mid q_2 \in (\opp{A'}\cap\qst{A'}) \setminus \ini{A'} \wedge q_1 = \pi^{-1} (q_2) \}
  \\ & \cup
  \{ a_2 \mapsto \icopya \seq \ileap {a_1}
      \mid a_2 \in \opp{A'}\cap\ans{A'} \wedge a_1 = \pi^{-1} (a_2) \}
  \\ & \cup
  \{ q_1 \mapsto \icopyq \seq \ileap {q_2}
      \mid q_1 \in \opp A\cap\qst A \wedge q_2 = \pi (q_1) \}
  \\ & \cup
  \{ a_1 \mapsto \icopya \seq \ileap {a_2}
      \mid a_1 \in \opp A\cap\ans A \wedge a_2 = \pi (a_1) \}
\end{align*}
This copycat engine is parametrised with an initial instruction $C$, which is
run when receiving an initial question. The engine for an ordinary copycat,
i.e. the identity of games, is $\copycat_{\icopyi,\pi,\mf A}$.
By slight abuse of notation, write $\copycat_{\mf A}$ for the
singleton copycat game net $\gnet {\singleton(\copycat_{\icopyi,\pi,\mf A})}{\mf A \fun \pi\cdot\mf {A}}$.

Following~\cite{DBLP:journals/entcs/GabbayG12}, we define a partial order $\leq$ over polarities, $\polarities$, as $ \lO \leq \lO, \lO \leq \lP, \lP \leq \lP $ and a preorder $\swapping$ over traces from $P_{\mf A}$ to be the least
  reflexive and transitive such that if $l_1 \leq l_2$ then  
\begin{multline*}
s_1 \snoc \lmsg{l_1}{a_1}{p_1}{p_1'}{d_1} \snoc \lmsg{l_2}{a_2}{p_2}{p_2'}{d_2} \snoc s_2 \\ \swapping
                s_1 \snoc \lmsg{l_2}{a_2}{p_2}{p_2'}{d_2} \snoc \lmsg{l_1}{a_1}{p_1}{p_1'}{d_1} \snoc s_2,
\end{multline*}
where $p_1'\neq p_2$. A set of traces $S \subseteq P_{\mf A}$ is \emph{saturated} if and only
  if, for $s,s' \in P_{\mf A}$, $s' \swapping s$ and $s \in S$ implies $s' \in
  S$.  If $S \subseteq P_{\mf A}$ is a set of traces, let $\sat(S)$ be the smallest
  saturated set of traces that contains $S$.

The usual definition of the copycat strategy (in the alternating and
single-threaded setting) as a set of traces is
\[
  \Copycat_{\mf A,\mf A'}\cond{st,alt} \defeq \{ s \in P_{\mf A \fun \mf {A'}}\cond{st,alt}
    \mid \forall s' \leq_{\text{even}} s. \: {s'}\dual \hereditary A \pnptreq s' \hereditary A' \}
\]
\begin{definition}
A set of traces $S_1$ is \emph{$\lP$-closed} with respect to a set of traces
$S_2$ if and only if $s' \in S_1 \cap S_2$ and $s = s' \snoc \lmsg \lP ap{p'}d
\in S_1$ implies $s \in S_2$.
\end{definition}
The intuition of $\lP$-closure is that if the trace $s'$ is
``legal'' according to $S_2$, then any outputs that can occur after $s'$ in
$S_1$ are also legal.

\begin{definition}\label{def:gimpl}
  We say that a GAM net $f$ \emph{implements} a set of traces $S$ if and
  only if $S \subseteq \denot f$ and $\denot f$ is $\lP$-closed with respect
  to $S$.
\end{definition}
  This is the form of the statements of correctness for game nets that we want;
  it certifies that the net $f$ can accommodate all traces in $S$ and,
  furthermore, that it only produces legal outputs when given valid inputs.

The main result of this section establishes the correctness of the GAM for copycat.

\begin{theorem} \label{thm:copycatimplements}
  ${\copycat_{\pi,\mf A}}$ implements
  $\Copycat_{\mf A, \pi\cdot\mf {A}}$.
\end{theorem}
This is a direct corollary of the
Lem.~\ref{thm:swappings},\ref{thm:unravelling},\ref{thm:copycatsubset},
\ref{thm:copycatpmoves}, and \ref{thm:copycatplegal} given below.

\input{copycat.tex}

\subsection{Composition}
The definition of composition in Hyland-Ong
games~\cite{DBLP:journals/iandc/HylandO00} is eerily similar to our definition of
trace composition, so we might expect HRAM net composition to correspond to it.
That is, however, only superficially true: the nominal setting that we are
using~\cite{DBLP:journals/entcs/GabbayG12} brings to light what happens to the
justification pointers in composition.

If $A$ is an interface, $s \in \traces_A$ and $X \subseteq \supp(A)$, we define
the \emph{reindexing deletion} operator $s \del X$
as follows, where $\tuple{s',\rho} = s \del X$ inductively:
\begin{align*}
  \emptytrace \del X          & \defeq \tuple{\emptytrace,\id{}} &
  \\ s \snoc \lmsg lap{p'}d \del X & \defeq \tuple{s' \snoc \lmsg la{\rho(p)}{p'}d,\rho} & \text{ if } a \notin X
  \\ s \snoc \lmsg lap{p'}d \del X & \defeq \tuple{s',\rho \cup \{p' \mapsto \rho(p)\}} & \text{ if } a \in X
\end{align*}
We write $s \del X$ for $s'$ when $s \del X = \tuple{s',\rho}$ in the following
definition:
\begin{definition}
  The \emph{game composition} of the sets of traces $S_1 \subseteq \traces_{A \fun B}$
  and $S_2 \subseteq \traces_{B' \fun C}$ with $\pi \vdash B \pneq B'$ is
  \[
    S_1 \gcompose S_2 \defeq
      \{s \del B \mid s \in \traces_{A \tensor B \tensor C}
        \wedge s \del C \in S_1
        \wedge \pi \cdot s^{*B} \del A \in S_2 \}
  \]
\end{definition}
Clearly we have $S_1 \compose S_2 \neq S_1 \gcompose S_2$ for
sets of traces $S_1$ and $S_2$, which
reinforces the practical problem in the beginning of this section.

Composition is constructed out of three copycat-like behaviours, as sketched in Fig.~\ref{fig:composition}
for a typical play at some types $A$,$B$ and $C$. As a trace in the nominal model, this  is:
\begin{multline*}
(q6, p0, p1)\snoc 
(q4, p1, p2)\snoc
(q3, p2, p3)\snoc\\
(q2, p1, p4)\snoc
(q1, p4, p5)\snoc
(q5, p1, p6)\snoc
(a5, p6)\snoc\\
(a1, p5)\snoc
(a2, p4)\snoc 
(a3, p3)\snoc
(a4, p2)\snoc
(a6, p1)
\end{multline*}
\begin{figure}
\begin{center}\begin{tikzpicture}[descr/.style={text=black}] 
  \matrix [column sep=1.0em]
  {
       \node(A){$(A$}; & \node(B){$B)$}; & \node(B'){$(B'$}; & \node(C){$C)$};
       & \node(A'){$(A'$}; & \node(C'){$C')$};
    \\                           &                       &                       &                       &                       & \node(q6){$q_6^\lO$};
    \\                           &                       &                       & \node(q4){$q_4^\lP$}; &                       &
    \\                           &                       & \node(q3){$q_3^\lO$}; &                       &                       &
    \\                           & \node(q2){$q_2^\lP$}; &                       &                       &                       &
    \\     \node(q1){$q_1^\lO$}; &                       &                       &                       &                       &
    \\                           &                       &                       &                       & \node(q5){$q_5^\lP$}; &
    \\                           &                       &                       &                       & \node(a5){$a_5^\lO$}; &
    \\     \node(a1){$a_1^\lP$}; &                       &                       &                       &                       &
    \\                           & \node(a2){$a_2^\lO$}; &                       &                       &                       &
    \\                           &                       & \node(a3){$a_3^\lP$}; &                       &                       &
    \\                           &                       &                       & \node(a4){$a_4^\lO$}; &                       &
    \\                           &                       &                       &                       &                       & \node(a6){$a_6^\lP$};
  \\ };
  \path[color=white]
    (A) edge node[descr]{$\Rightarrow$} (B)
    (B) edge node[descr]{$\tensor$} (B')
    (B') edge node[descr]{$\Rightarrow$} (C)
    (C) edge node[descr]{$\rightarrow$} (A')
    (A') edge node[descr]{$\Rightarrow$} (C')
    ;
  \path[->] 
    (q4) edge [bend right=5] (q6)
    (q3) edge (q4)
    (q2) edge [bend left=30] (q6)
    (q1) edge (q2)
    (q5) edge (q6)
    (a5) edge [bend right=90] (q5)
    (a1) edge (q1)
    (a2) edge (q2)
    (a3) edge (q3)
    (a4) edge (q4)
    (a6) edge (q6)
  ;
  \path[dashed,->] 
     (q4) edge [bend left=15] (q6)
     (q2) edge (q3)
     (q5) edge (q1)
     (a1) edge (a5)
     (a3) edge (a2)
     (a6) edge (a4)
  ;
  \path[dotted,->] 
    (q2) edge [bend right=15] (q6)
  ;
\end{tikzpicture}\end{center}
\vspace{-.5cm}
\caption{Composition from copycat}
\label{fig:composition}
\end{figure}
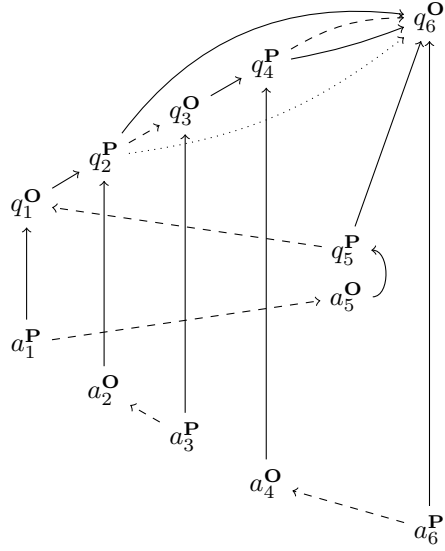
We see that this \emph{almost} corresponds to three interleaved copycats as
described above; between $A,B,C$ and $A',B',C'$. There is, however, a small
difference: The move $q_1$, if it were to blindly follow the recipe of a
copycat, would dereference the pointer $p_4$, yielding $p_3$, and so
incorrectly make the move $q_5$ justified by $q_3$, whereas it really should be
justified by $q_6$ as in the diagram. This is precisely the problem explained at the beginning of this section.

To make a pointer \emph{extension}, when the $B$-initial move $q_3$
is performed, it should map $p_4$ not only to $p_3$, but also to the
pointer that $p_2$ points to, which is $p_1$ (the dotted line in the diagram).
When the A-initial move $q_1$ is performed, it has access to both of these pointers
that $p_4$ maps to, and can correctly make the $q_5$ move by associating it with pointers $p_1$
and a fresh $p_6$.

Let $\mf {A'}$, $\mf {B'}$, and $\mf {C'}$ be game interfaces such that
$\pi_{\mf A} \vdash \mf A \pneq \mf{A'}$,
$\pi_{\mf B} \vdash \mf B \pneq \mf{B'}$,
$\pi_{\mf C} \vdash \mf C \pneq \mf{C'}$,
and
\begin{align*}
     \engine {A' \fun A} {P_A} & = \copycat_{\iextq,\pi^{-1}_{\mf A},\mf A'}
  \\ \engine {B \fun B'} {P_B} & = \copycat_{\iexti,\pi_{\mf B},\mf B}
  \\ \engine {C \fun C'} {P_C} & = \copycat_{\icopyi,\pi_{\mf C},\mf C}, \text{ where }
  \\ \iexti  & \defeq 0,3 \leftarrow \niget\ 0 \seq 1 \leftarrow \nialloc\ 1,0
  \\ \iextq  & \defeq \emptyset,0 \leftarrow \niget\ 0 \seq 1 \leftarrow \nialloc\ 1,3
\end{align*}
Then the game composition operator $\gcomp_{\mf A,\mf B,\mf C}$ 
is:
\[\gcomp_{\mf A,\mf B,\mf C} \defeq \engine {(A \fun B) \tensor (B' \fun C) \fun (A' \fun C')} {P_A \cup P_B \cup P_C}.
\]

Using the game composition operator $\gcomp$ we can define GAM-net composition using \netcat compact closed combinators. Let $f:\mf A\Rightarrow \mf B, g:\mf B\Rightarrow \mf C$ be GAM-nets. Then their composition is defined as
\begin{align*}
&f;_{GAM}g \defeq 
\Lambda^{-1}_{A}(\Lambda_A(f)\tensor\Lambda_B(g)); \gcomp_{\mf A,\mf B,\mf C})),
\text{where}\\
&\Lambda_A(f:A\rightarrow B)\defeq(\eta_A;(\id {A\dual}\tensor f)):I\rightarrow A\dual\tensor B\\
&\Lambda^{-1}_A(f:I\rightarrow A\tensor B)\defeq((\id A \tensor f);(\varepsilon_A\tensor \id B)):A\rightarrow B.
\end{align*}
Composition is represented diagrammatically as in Fig.~\ref{fig:cgam}.
\begin{figure}
\begin{center}

  \begin{tikzpicture}
    \node[scale=0.9]{
      \input{k.latex}
    };
  \end{tikzpicture}

\end{center}
\caption{Composing GAMs using the $K$ HRAM}
\label{fig:cgam}
\end{figure}
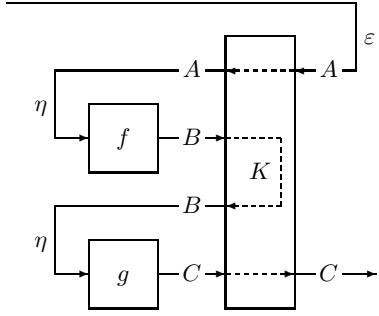
Note the comparison with the naive composition from Fig.~\ref{fig:nlhc}. HRAMs $f$ and $g$ are not plugged in directly, although the interfaces match. Composition is mediated by the operator $\gcomp$, which preserves the locality of freshly generated names, exchanging non-local pointer names with local pointer names and storing the mapping between the two as copy-cat links, indicated diagrammatically by dotted lines in $K$. 
\begin{theorem} \label{thm:Kimplements}
  If $f \hastype \mf {A \tfun B}$ and $g \hastype \mf {B' \tfun C}$ are game
  nets such that $\pi_{\mf B} \vdash \mf B \pneq \mf B'$,
  $f$ implements $S_f \subseteq P_{\mf {A \fun B}}$, and
  $g$ implements $S_g \subseteq P_{\mf {B' \fun C}}$, then
  ${f \compose_{GAM} g}$ implements $(S_f \gcompose S_g)$.
\end{theorem}

\input{k.tex}

\subsection{Diagonal}
For game interfaces $\mf A_1,\mf A_2,\mf A_3$ and permutations $\pi_{ij}$
such that $\pi_{ij} \vdash \mf A_i \pneq \mf A_j$ for $i\neq j\in\{1,2,3\}$,
we define the family of diagonal engines as:
\[
  \delta_{\pi_{12},\pi_{13},\mf A} = (A_1 \fun A_2 \tensor A_3, P_1\otimes P_2\otimes P_3)
\]
where, for $i\in\{2,3\}$, 
\begin{align*}
  P_1 & \defeq
  \{ q_1 \mapsto \icopyq \seq \niifzero\ 3\ {(\ileap {q_2})}\ {(\ileap {q_3})}
      \\&\qquad\mid q_1 \in \opp {A_1}\cap\qst {A_1} \wedge q_2 = \pi_{12} (q_1)
                                            \wedge q_3 = \pi_{13} (q_1)\}
  \\ & \cup
  \{ a_1 \mapsto \icopya \seq \niifzero\ 3\ {(\ileap {a_2})}\ {(\ileap {a_3})}
      \\&\qquad\mid a_1 \in \opp {A_1}\cap\ans {A_1} \wedge a_2 = \pi_{12} (a_1)
                                            \wedge a_3 = \pi_{13} (a_1)\}\\
  P_i & \defeq
  \{ q_i \mapsto 3 \leftarrow \iset (i-2) \seq \icopyi \seq \ileap {q_1}
  \mid q_i \in \ini {A_i} \wedge q_1 = \pi_{1i}^{-1} (q_i) \}
   \\ & \cup
  \{ q_i \mapsto \icopyq \seq \ileap {q_1}
      \mid q_i \in (\opp{A_i}\cap\qst{A_i}) \setminus \ini{A_i} \wedge q_1 = \pi_{1i}^{-1} (q_i) \}
  \\ & \cup
  \{ a_i \mapsto \icopya \seq \ileap {a_1}
      \mid a_i \in \opp{A_i}\cap\ans{A_i} \wedge a_1 = \pi_{1i}^{-1} (a_i) \}.
\end{align*}
The diagonal is almost identical to the copycat, except that an integer value of 0 or 1 is associated, in the heap,
with the name of each message arriving on the $A_2$ and $A_3$ interfaces (hence the \texttt{set} statements, to be used for routing back messages arriving on $A_1$ using \texttt{ifzero} statements).
By abuse of notation, we also write $\delta$ for the net $\singleton(\delta)$. 
\begin{lemma}\label{lem:prog}
  The $\delta$ net is the diagonal net, i.e.\ 
  $\denot{\delta_{\pi_{12},\pi_{23},\mf A} \compose \Pi_{i}}
    = \denot{\copycat_{\pi_i,\mf A }}.$
\end{lemma}
\begin{proof}
  We show that $s \in \denot{\delta_{\pi_{12},\pi_{23}} \compose \Pi_1}$
  implies $s \in \denot{\copycat_{\pi_{12}, \mf A_1,
  \mf A_2}}$ and the converse (the $\Pi_2$ case is analogous), by induction
  on the trace length. There is a simple relationship between the heap
  structures of the respective net configurations --- they have the same
  structure but the diagonal stores additional integers for identifying what
  ``side'' a move comes from.
\end{proof}

\subsection{Fixpoint}\label{sec:fix}

We define a family of GAMs $\mathit{Fix}_{\mf A}$ with interfaces $(\mf A_1\Rightarrow \mf A_2)\Rightarrow \mf A_3$ where there exist permutations $\pi_{i,j}$ such that $\pi_{i,j} \vdash \mf A_i \pneq \mf A_j$ for $i\neq j\in\{1,2,3\}$. The fixpoint engine is defined as $\mathit{Fix}_{\pi_{12},\pi_{13},\mf A}=\Lambda^{-1}_A(\delta_{\pi_{12},\pi_{13},\mf A})$.

Let $\mathit{fix}_{\pi_{12},\pi_{13},\mf A}:(\mf A\Rightarrow \pi_{12}\cdot\mf A)\Rightarrow \pi_{13}\cdot\mf A$ be the game-semantic strategy for fixpoint in Hyland-Ong games~\cite[p.~364]{DBLP:journals/iandc/HylandO00}.

\begin{theorem}
$\mathit{Fix}_{\pi_{12},\pi_{13},\mf A}$ implements $\mathit{fix}_{\pi_{12},\pi_{13},\mf A}$.
\end{theorem}

The proof of this is immediate considering the three cases of moves from the
definition of the game-semantic strategy. It is interesting to note here that
we ``\emph{force}'' a HRAM with interface $A_1\Rightarrow A_2\otimes A_3$ into
a GAM with game interface $(\mf A_3\Rightarrow \mf A_1)\Rightarrow \mf A_2$,
which has underlying interface $( A_3\Rightarrow  A_1)\Rightarrow  A_2$. In the
\netcat category, which is symmetric compact-closed, the two
interfaces are isomorphic (with $A_1\dual\otimes A_2\otimes A_3$), but as
game interfaces they are not. It is rather surprising that we can reuse our
diagonal GAMs in such brutal fashion: in the game interface for fixpoint there
is a reversed enabling relation between $A_3$ and $A_1$. The reason why this
still leads to legal plays only is because the onus of producing the
justification pointers in the initial move for $A_3$ lies with the Opponent,
which cannot exploit the fact that the diagonal is ``wired illegally''. It only
sees the fixpoint interface and must play accordingly. It is fair to say that
that fixpoint interface is more restrictive to the Opponent than the diagonal
interface, because the diagonal interface allows extra behaviours, e.g.\
sending initial messages in $A_3$, which are no longer legal.

\subsection{Other ICA constants}\label{sec:oica}

A GAM net for an integer literal $n$ can be defined using the following
engine (whose interface corresponds to the ICA $\expp$ type).
\begin{align*}
  \mathname{lit}_n & \defeq \engine {\{\port \lO q,\port \lP a\}} P\text{, where}
  \\ P & \defeq \{ q \mapsto \niflip\ 0,1 \seq 1 \leftarrow \niset\ \emptyset
      \seq 2 \leftarrow \niset\ n \seq \nispark\ a
      \}
\end{align*}
We see that upon getting an input question on port $q$, this engine will
respond with a legal answer containing $n$ as its value (register $2$).

The conditional at type $\expp$ can be defined using the following engine,
with the convention that $\{\port \lO {q_i},\port \lP {a_i}\} = \expp_i$.
\begin{align*}
  \mathname{if}  \defeq & \engine {\expp_1 \fun \expp_2 \fun \expp_3 \fun \expp_4} P\text{, where}
  \\ P  \defeq & \{ q_4 \mapsto \icopyi \seq \nispark\ q_1,
  \\           &  a_1 \mapsto \icopya \seq \niflip\ 0,1 \seq \icopyi \seq 
                      \niifzero\ 2\ (\nispark\ q_3)\ (\nispark\ q_2),
  \\           &  a_2 \mapsto \icopya \seq \nispark\ a_4,
  \\           & a_3 \mapsto \icopya \seq \nispark\ a_4 \}
\end{align*}
We can also define primitive operations, e.g. $+ \hastype \expp \fun \expp \fun \expp$,
in a similar manner. An interesting engine is that for $\mathname{newvar}$:
\begin{align*}
\mathname{newvar} \defeq & \engine {(\expp_1 \tensor (\expp_2 \fun \comm_3) \fun \expp_4) \fun \expp_5}P
\\ P \defeq & \{q_5 \mapsto 3 \leftarrow \niset\ 0 \seq \icopyi \seq \nispark\ q_4,
\\          &  q_1 \mapsto \emptyset,2 \leftarrow \niget\ 0 \seq \iflip\ 0,1 \seq 1 \leftarrow \niset\ \emptyset \seq \nispark\ a_1,
\\          &  q_3 \mapsto \niflip\ 0,1 \seq 1 \leftarrow \nialloc\ 0,1 \seq \nispark\ q_2,
\\          &  a_2 \mapsto \emptyset,3 \leftarrow \niget\ 0 \seq \niupdate\ 3\ 2 \seq \icopya \seq \nispark\ a_3,
\\          &  a_4 \mapsto \icopya \seq \nispark\ a_5 \}
\end{align*}
We see that we store the variable in the second component of the
justification pointer that justifies $q_4$, so that it can be accessed in
subsequent requests. A slight problem is that moves in $\expp_2$ will actually
not be justified by this pointer which we remedy in the $q_3$ case, by storing
a pointer to the pointer with the variable as the second component of the
justifier of $q_2$, which means that we can access and update the variable in $a_2$.

We can easily extend the HRAMs with new instructions to interpret parallel
execution and semaphores, but we omit them from the current presentation.
\section{Seamless distributed compilation for ICA}\label{sec:disc}

\subsection{The language ICA}

ICA is PCF extended with constants to facilitate local effects. Its ground
types are expressions and commands ($\mathsf{exp, com}$), with the type of
assignable variables desugared as $\mathsf{var\defeq exp\times (exp\rightarrow
com)}$. Dereferencing and assignment are desugared as the first, respectively
second, projections from the type of assignable variables. The local variable
binder is $\mathsf{new:(var\rightarrow com)\rightarrow com}$. ICA also has a
type of split binary semaphores $\mathsf{sem\defeq com\times com}$, with the
first and second projections corresponding to $\mathsf{set, get}$, respectively
(see~\cite{DBLP:journals/apal/GhicaM08} for the full definition, including the
game-semantic model). 

In this section we give a compilation method for ICA into GAM nets. The
compilation is compositional on the syntax and it uses the constructs of the
previous section. ICA types are compiled into GAM interfaces which correspond
to their game-semantic arenas in the obvious way. We will use $A, B,\ldots$ to
refer to an ICA type and to the GAM interface. Sec.~\ref{sec:gamn} has already
developed all the infrastructure needed to interpret the constants of ICA
(Sec.~\ref{sec:oica}), including fixpoint (Sec.~\ref{sec:fix}). Given an ICA
type judgment $\Gamma\vdash M:A$ with $\Gamma$ a list of variable-type
assignments $x_i:A_i$, $M$ a term and $A$ a type, a GAM implementing it $G_M$
is defined compositionally on the syntax as follows:
\begin{align*}
G_{\Gamma\vdash MM':A}&=\delta_{\pi_1,\pi_2,\Gamma};_{\mathit{GAM}}(G_{\Gamma\vdash M:A\rightarrow B}\otimes 
G_{\Gamma\vdash M':B});_{\mathit{GAM}}\geval_{A,B}\\
G_{\Gamma\vdash\lambda x:A.M:A\rightarrow B} &= \Lambda_A(G_{\Gamma,x:A\vdash M:B} )\\
G_{x:A,\Gamma\vdash x:A} &= \Pi_{\mf G A} \compose \copycat_{A,\pi},
\end{align*}
Where $\geval_{A,B}\defeq\Lambda^{-1}_{B}(\copycat_{A\Rightarrow B,\pi})$ for a
suitably chosen port renaming $\pi$ and $\Pi_{\mf G A}$ and $\Pi_{\mf G 1}$ and
$\Pi_{\mf G 2}$ are HRAMs with signatures $\Pi_{\mf G i}=(A_1\tensor A_2\fun
A_3,P_i)$ such that they copycat between $A_3$ and $A_i$ and ignore
$A_{j\neq i}$.
 The
interpretation of function application, which is the most complex, is shown
diagrammatically in Fig.~\ref{fig:diagap}. The copycat connections are shown
using dashed lines.
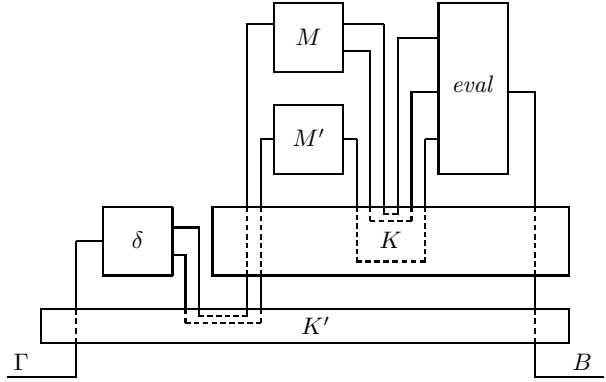
\begin{figure}
\centering

  \begin{tikzpicture}
    \node[scale=0.9]{
      \input{app.latex}
    };
  \end{tikzpicture}

\caption{GAM net for application}
\label{fig:diagap}
\end{figure} 
\begin{theorem}
If $M$ is an ICA term, $G_M$ is the GAM implementing it and $\sigma_M$ its game-semantic strategy then $G_M$ implements $\sigma_M$. 
\end{theorem}
The correctness of compilation follows directly from the correctness of the individual GAM nets and the correctness of GAM composition $;_{\mathit{GAM}}$.

\subsection{Prototype implementation}

Following the recipe in the previous section we can produce an implementation of any ICA term as a GAM net. GAMs are just special-purpose HRAMs, with no special operations. HRAMs, in turn, can easily be implemented on any conventional computer with the usual store, control and communication facilities. A GAM net is also just a special-purpose HRAM net, which is a powerful abstraction of communication processes, as it subsumes through the \texttt{spark} instruction communication between processes (threads) on the same physical machine or located on distinct physical machines and communicating via a point-to-point network. We have built a prototype compiler based on GAMs by implementing them in C, managing processes using standard UNIX threads and physical network distribution using MPI~\cite{gropp1999using}.\footnote{Download with source code from \url{http://veritygos.org/gams}.}

The actual distribution is achieved using light pragma-like code annotations. In order to execute a program at node $A$ but delegate one computation to node $B$ and another computation to node $C$ we simply annotate an ICA program with node names, e.g.:
\begin{verbatim}
  {new x. x := {f(x)}@B + {g(x)}@C; !x}@A
\end{verbatim}
Note that this gives node $B$, via function $f$, read-write access to memory location $x$ which is located at node $A$. Accessing non-local resources is possible, albeit possibly expensive. 

\begin{figure}
\centering

  \begin{tikzpicture}
    \node[scale=0.9]{
      \input{appopt.latex}
    };
  \end{tikzpicture}

\caption{Optimised GAM net for application}
\label{fig:diagapopt}
\end{figure}
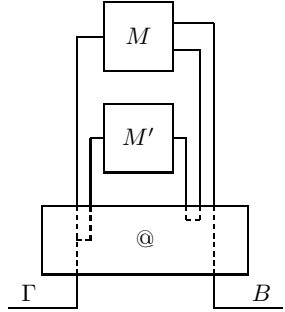 
Several facts make the compilation process quite remarkable:
\begin{itemize}
\item It is \emph{seamless} (in the sense of~\cite{GhicaF12}), allowing distributed compilation where communication is never explicit but always realised through function calls. 
\item It is \emph{flexible}, allowing any syntactic sub-term to be located at any designated physical location, with no impact on the semantics of the program. The access of \emph{non-local} resources is always possible, albeit possibly at a cost (latency, bandwidth, etc.).
\item It is \emph{dynamic}, allowing the relocation of GAMs to different physical nodes at run time. This can be done with extremely low overhead if the GAM heap is empty. 
\item It does not require any form of \emph{garbage collection}, even on local nodes, although the language combines (ground) state, higher-order functions and concurrency. This is because a pointer associated with a pointer is not needed if and only if the question is answered; then it can be safely deallocated. 
\end{itemize}

The current implementation does not perform any optimisations, and the
resulting code is inefficient. Looking at the implementation of application in
Fig.~\ref{fig:diagap} it is quite clear that a message entering the GAM net via
port $A$ needs to undergo four pointer renamings before reaching the GAM for
$M$. This is the cost we pay for compositionality. However, the particular
configuration for application can be significantly simplified using standard
peephole optimisation, and we can reach the much simpler, still correct
implementation in Fig.~\ref{fig:diagapopt}. Here the functionality of the two
compositions, the diagonal, and the $\geval$ GAMs have been combined and optimised
into a single GAM, requiring only one pointer renaming before reaching $M$.
Other optimisations can be
introduced to simplify GAM nets, in particular to obviate the need for the use
of composition GAMs $K$, for example by showing that composition of first-order
closed terms (such as those used for most constants) can be done directly.

\section{Conclusions, further work}

In a previous paper we have argued that distributed and heterogeneous programming would benefit from the existence of architecture-agnostic, seamless compilation methods for conventional programming languages which can allow the programmer to focus on solving algorithmic problems without being overwhelmed by the minutiae of driving complex computational systems~\cite{GhicaF12}. In \emph{loc. cit.} we give such a compiler for PCF, based directly on the Geometry of Interaction. In this paper we show how Game Semantics can be expressed operationally using abstract machines very similar to networked conventional computers, a further development of the IAM/JAM game machines. We believe any programming language with a semantic model expressed as Hyland-Ong-style pointer games~\cite{DBLP:journals/iandc/HylandO00} can be readily represented using GAMs and then compiled to a variety of platforms such as MPI. Even more promising is the possible leveraging of more powerful infrastructure for distributed computing that can mask much of the complexities of distributed programming, such as fault-tolerance~\cite{murray2011ciel}.

The compositional nature of the compiler is very important because it gives
rise to a very general notion of foreign-function interface, expressible both
as control and as communication, which allows a program to interface with other
programs, in a syntax-independent way (see~\cite{DBLP:conf/memocode/Ghica11}
for a discussion), opening the door to the seamless development of
heterogeneous open systems in a distributed setting.

We believe we have established a solid foundational platform on which to build realistic seamless distributed compilers. Further work is needed in optimising the output of the compiler which is currently, as discussed, inefficient. The sources of inefficiency in this compiler are not just the generation of heavy-duty plumbing, but also the possibly unwise assignment of computation to nodes, requiring excessive network communication. Previous work in game semantics for resource usage can be naturally adapted to the operational setting of the GAMs and facilitate the automation of optimised task assignment~\cite{DBLP:conf/popl/Ghica05}.

\bibliographystyle{abbrv}
\bibliography{\jobname,dblp}
\end{document}

%% file: macros.tex
\usepackage{amsmath, amssymb, amsthm}
\usepackage{stmaryrd}
\usepackage{bussproofs}
\usepackage{graphicx}
\usepackage{mathtools}
\usepackage{framed}
\usepackage{tikz}
\usepackage{url}
\usepackage{bbm}
\usetikzlibrary{matrix,arrows,decorations.pathmorphing}

\newcommand{\app}{\ensuremath{\:}}
\newcommand{\mathname}[1]{\ensuremath{\text{\emph{#1}}}}
\newcommand{\ensuretext}[1]{\ifmmode \text{#1} \else #1 \fi}
\newcommand{\boldg}[1]{\ifmmode \mathbf{#1} \else \textbf{#1} \fi}
\newcommand{\keyword}[1]{\boldg{#1}}
\newcommand{\tuple}[1]{\ensuremath{(#1)}}
\newcommand{\defeq}{\ensuremath{\stackrel{\Delta}{=}}}

\newtheorem{theorem}{Theorem}[section]
\newtheorem{lemma}[theorem]{Lemma}
\newtheorem{proposition}[theorem]{Proposition}

\newtheorem{definition}[theorem]{Definition}

\newcommand*{\textcal}[1]{%
      \textit{\fontfamily{qzc}\selectfont#1}%
}

\newcommand{\partialfun}{\ensuremath{\rightharpoonup}}
\newcommand{\many}[1]{\ensuremath{\overline{#1}}}

\newcommand{\cod}{\mathrm{cod}}
\newcommand{\List}{\ensuremath{\text{\textbf{List}}}}

\newcommand{\pset}{\ensuremath{\mathcal P}}
\newcommand{\psetfin}{\ensuremath{\pset_{\mathname{fin}}}}
\newcommand{\msetfin}{\ensuremath{\textbf {Mset}_{\mathname{fin}}}}
\newcommand{\tfun}{\ensuremath{\rightarrow}}
\newcommand{\tprod}{\ensuremath{\times}}

\newcommand{\compose}{\ensuremath{;}}
\newcommand{\hastype}{\ensuremath{\app : \app}}

\newcommand{\portnames}{\ensuremath{\mathbb A}}
\newcommand{\pointers}{\ensuremath{\mathbb P}}
\newcommand{\ports}{\ensuremath{\text{\textcal{Port}}}}
\newcommand{\at}[1]{\ensuremath{^{(\const{#1})}}}
\newcommand{\lO}{\const O}
\newcommand{\lP}{\const P}

\newcommand{\const}[1]{\ensuremath{\keyword{#1}}}
\newcommand{\dual}{\ensuremath{^{*}}}
\newcommand{\port}[2]{\tuple{#1,#2}}
\newcommand{\supp}{\mathname{sup}}
\newcommand{\tensor}{\ensuremath{\otimes}}
\newcommand{\fun}{\ensuremath{\Rightarrow}}
\newcommand{\interfaces}{\ensuremath{\mathcal I}}
\newcommand{\pneq}{\ensuremath{=_\portnames}}
\newcommand{\pnptreq}{\ensuremath{=_{\portnames\pointers}}}

\newcommand{\codes}{\ensuremath{\mathcal C}}
\newcommand{\instrs}{\ensuremath{\text{\textcal{Instr}}}}
\newcommand{\datas}{\ensuremath{\mathcal D}}
\newcommand{\integers}{\ensuremath{\mathbb Z}}
\newcommand{\engine}[2]{\tuple{#1,#2}}
\newcommand{\engines}{\ensuremath{\mathcal E}}
\newcommand{\engineconf}[2]{\tuple{#1,#2}}
\newcommand{\engineconfs}{\ensuremath{\mathcal K}}
\newcommand{\enginestep}[3]{\ensuremath{\xrightarrow[#1,#2]{#3}}}

\newcommand{\hmap}[1]{\ensuremath{\mapsto \tuple{#1}}}
\newcommand{\initial}{\mathname{initial}}
\newcommand{\singleton}{\mathname{singleton}}

\newcommand{\instr}[1]{\ensuremath{\text{\texttt{#1}}}}
\newcommand{\seq}{\ensuremath{; \app}}
\newcommand{\ileap}{\instr{spark} \app}
\newcommand{\iend}{\instr{end}}
\newcommand{\icopyi}{\instr{cci}}
\newcommand{\icopyq}{\instr{ccq}}
\newcommand{\icopya}{\instr{cca}}
\newcommand{\iexti}{\instr{exi}}
\newcommand{\iextq}{\instr{exq}}
\newcommand{\iflip}{\instr{flip}}

\newcommand{\iset}[1]{\ensuremath{\instr{set} \app #1}}

\newcommand{\nialloc}{\instr{new}}
\newcommand{\niget}{\instr{get}}
\newcommand{\nifree}{\instr{free}}
\newcommand{\niset}{\instr{set}}
\newcommand{\niflip}{\instr{flip}}
\newcommand{\regupd}[2]{\ensuremath{[#1 := #2]}}
\newcommand{\niifzero}{\instr{ifzero}}
\newcommand{\nispark}{\instr{spark}}
\newcommand{\niupdate}{\instr{update}}

\newcommand{\rmsg}[2]{\tuple{#1,#2}}

\newcommand{\thread}[2]{\tuple{#1,#2}}

\newcommand{\msgf}{\mathname{msg}}
\newcommand{\regsf}{\mathname{regs}}

\newcommand{\msgs}{\ensuremath{\mathcal M}}
\newcommand{\msg}[4]{\tuple{#1,#2,#3,#4}}
\newcommand{\lmsg}[5]{\tuple{#1,\msg{#2}{#3}{#4}{#5}}}
\newcommand{\net}[3]{\tuple{#1,#2,#3}}
\newcommand{\nets}{\ensuremath{\mathcal S}}
\newcommand{\netconf}[2]{\tuple{#1,#2}}
\newcommand{\netconfs}{\ensuremath{\mathcal N}}
\newcommand{\netstep}[1]{\ensuremath{\xrightarrow{#1}}} 
\newcommand{\conn}{\ensuremath{\chi}}
\newcommand{\rthread}[2]{\tuple{#1,#2}}

\newcommand{\denot}[1]{\ensuremath{\llbracket #1 \rrbracket}}
\newcommand{\traces}{\mathname{traces}}
\newcommand{\ndel}{\ensuremath{{-}}}
\newcommand{\del}{\ensuremath{\downharpoonright}}
\newcommand{\emptytrace}{\ensuremath{\epsilon}}
\newcommand{\snoc}{\ensuremath{{::}}}

\newcommand{\netcat}{\keyword{HRAMnet}}
\newcommand{\id}[1]{\ensuremath{\mathname{id}_{#1}}}
\newcommand{\assocempty}{\ensuremath{\alpha}}
\newcommand{\assoc}[3]{\ensuremath{\assocempty_{#1,#2,#3}}}

\newcommand{\lunitor}[1]{\ensuremath{\lambda_{#1}}}
\newcommand{\runitor}[1]{\ensuremath{\rho_{#1}}}
\newcommand{\commutempty}{\ensuremath{\gamma}}
\newcommand{\commut}[2]{\ensuremath{\commutempty_{#1,#2}}}
\newcommand{\unit}[1]{\ensuremath{\eta_{#1}}}
\newcommand{\counit}[1]{\ensuremath{\varepsilon_{#1}}}

\newcommand{\mf}[1]{\ensuremath{\mathfrak{#1}}}
\newcommand{\ginterface}[4]{\tuple{#1,#2,#3,#4}}
\newcommand{\ginterfaces}{\interfaces_{\mf G}}
\newcommand{\qst}[1]{\ensuremath{\mathname{qst}_{\mf{#1}}}}
\newcommand{\ans}[1]{\ensuremath{\mathname{ans}_{\mf{#1}}}}
\newcommand{\ini}[1]{\ensuremath{\mathname{ini}_{\mf{#1}}}}
\newcommand{\opp}[1]{\ensuremath{\mathname{opp}_{\mf{#1}}}}
\newcommand{\prop}[1]{\ensuremath{\mathname{prop}_{\mf{#1}}}}
\newcommand{\enables}[1]{\ensuremath{\vdash_{\mf{#1}}}}
\newcommand{\com}{\const{com}}
\newcommand{\cond}[1]{^{\mathname{#1}}}

\newcommand{\polarities}{\ensuremath{L}}
\newcommand{\sat}{\mathname{sat}}

\newcommand{\gnet}[2]{\tuple{#1,#2}}
\newcommand{\enabled}{\mathname{enabled}}
\newcommand{\cp}{\mathname{cp}}
\newcommand{\fp}{\mathname{fp}}
\newcommand{\ptrs}{\mathname{ptrs}}
\newcommand{\gcompose}{\ensuremath{;_{\mf G}}}
\newcommand{\hereditary}{\ensuremath{\upharpoonright}}

\newcommand{\swapping}{\ensuremath{\preccurlyeq}}

\newcommand{\gcomp}{\ensuremath{K}}
\newcommand{\copycat}{C\!\!\!C}
\newcommand{\Copycat}{c\!c}

\newcommand{\ready}{\mathname{ready}}
\newcommand{\geval}{\mathname{eval}}

\newcommand{\sink}{\ensuremath{!}}

\newcommand{\expp}{\ensuremath{\mathtt{exp}}}
\newcommand{\comm}{\ensuremath{\mathtt{com}}}
\newcommand{\bijpi}{\tilde{\pi}}
\newcommand{\pilmsg}[5]{\ensuremath{
  \lmsg {#1}{\bijpi_\portnames(#2)}{\bijpi_\pointers(#3)}{\bijpi_\pointers(#4)}{#5}
  }}

%% file: hramop.tex
\begin{figure} \label{fig:hramop}
\small
\begin{framed}
\[
  \engineconf {\rthread {i \leftarrow \nialloc\ j,k \seq C}{\many d}
               \cup \many t
               }
               h
  \enginestep E \conn {}
  \engineconf {\rthread C{\many d \regupd i p}
               \cup \many t
               }
               {h \cup \{p \hmap {d_j,d_k}\}}
\text{ if $p \notin \supp(h)$}
\]

\[
  \begin{split}
  \engineconf {\rthread {i,j \leftarrow \niget\ k \seq C}{\many d}
               \cup \many t
               }
               {h \cup \{d_k \hmap {d, d'} \}}
  \enginestep E \conn {} \\
  \engineconf {\rthread C{\many d \regupd i d \regupd j {d'}}
               \cup \many t
               }
               {h \cup \{d_k \hmap {d,d'}\}}
  \end{split}
\]

\[
  \begin{split}
  \engineconf {\rthread {\niupdate\ i,j \seq C}{\many d}
               \cup \many t
               }
               {h \cup \{d_i \hmap {d, d'} \}}
  \enginestep E \conn {} \\
  \engineconf {\rthread C{\many d \regupd i d \regupd j {d'}}
               \cup \many t
               }
               {h \cup \{d_i \hmap {d,d_j}\}}
  \end{split}
\]

\[
  \engineconf {\rthread {\nifree\ i \seq C}{\many d}
               \cup \many t
               }
               {h \cup \{d_i \hmap {d, d'} \}}
  \enginestep E \conn {}
  \engineconf {\rthread C{\many d \regupd i \emptyset}
               \cup \many t
               }
               h
\]

\[
  \engineconf {\rthread {\niflip\ i,j \seq C}{\many d}
               \cup \many t
               }
               h
  \enginestep E \conn {}
  \engineconf {\rthread C{\many d \regupd i {d_j} \regupd j {d_j}}
               \cup \many t
               }
               h
\]

\[
  \engineconf {\rthread {i \leftarrow \niset\ j \seq C}{\many d}
               \cup \many t
               }
               h
  \enginestep E \conn {}
  \engineconf {\rthread C{\many d \regupd i j}
               \cup \many t
               }
               h
\]

\[
  \engineconf {\rthread {\niifzero\ i\ {c_1}\ {c_2} \seq C}{\many d \regupd i 0}
               \cup \many t
               }
               h
  \enginestep E \conn {}
  \engineconf {\rthread {c_1}{\many d \regupd i \emptyset}
               \cup \many t
               }
               h
\]

\[
  \engineconf {\rthread {\niifzero\ i\ {c_1}\ {c_2} \seq C}{\many d \regupd i {n+1}}
               \cup \many t
               }
               h
  \enginestep E \conn {}
  \engineconf {\rthread {c_2}{\many d \regupd i \emptyset}
               \cup \many t
               }
               h
\]

\[
  \engineconf {\rthread {\nispark\ a}{\many d}
               \cup \many t
               }
               h
  \enginestep E \conn {\rmsg {\conn(a)}{\msgf(\many d)}}
  \engineconf {\many t}
               h
\text{ if $\port \lO {\conn(a)} \notin A$}
\]

\[
  \engineconf {\rthread {\nispark\ a}{\many d}
               \cup \many t
               }
               h
  \enginestep E \conn {}
  \engineconf {\rthread {P(\conn(a))}{\regsf(\msgf(\many d))}
               \cup \many t
               }
               h
\text{ if $\port \lO {\conn(a)} \in A$}
\]

\[
  \engineconf {\many t}
               h
  \enginestep E \conn {{\rmsg a{\many d}}^\bullet}
  \engineconf {\rthread {P(a)}{\regsf(\many d)}
               \cup \many t
               }
               h
\text{ if $\port \lO a \in A$}
\]

\[
  \engineconf {\rthread {\iend}{\many d}
               \cup \many t
               }
               h
  \enginestep E \conn {}
  \engineconf {\many t}
               h
\]
\end{framed}
\caption{Operational semantics of HRAMs}
\end{figure}

%% file: hramnet.latex
\setlength{\unitlength}{4144sp}%
\begingroup\makeatletter\ifx\SetFigFont\undefined%
\gdef\SetFigFont#1#2#3#4#5{%
  \reset@font\fontsize{#1}{#2pt}%
  \fontfamily{#3}\fontseries{#4}\fontshape{#5}%
  \selectfont}%
\fi\endgroup%
\begin{picture}(2364,1644)(169,-1873)
\thinlines
{\color[rgb]{0,0,0}\put(1441,-511){\line( 1, 0){1080}}
}%
{\color[rgb]{0,0,0}\put(1441,-1051){\line( 1, 0){1080}}
}%
{\color[rgb]{0,0,0}\put(1441,-1321){\framebox(1080,1080){}}
}%
{\color[rgb]{0,0,0}\multiput(1441,-781)(58.37838,0.00000){19}{\line( 1, 0){ 29.189}}
}%
{\color[rgb]{0,0,0}\put(1981,-511){\line( 0,-1){540}}
}%
{\color[rgb]{0,0,0}\multiput(2251,-511)(0.00000,-60.00000){5}{\line( 0,-1){ 30.000}}
}%
{\color[rgb]{0,0,0}\multiput(1711,-511)(0.00000,-60.00000){5}{\line( 0,-1){ 30.000}}
}%
{\color[rgb]{0,0,0}\put(991,-1321){\line( 0,-1){ 90}}
\put(991,-1411){\line( 1, 0){720}}
\put(1711,-1411){\line( 0, 1){ 90}}
}%
{\color[rgb]{0,0,0}\put(451,-1321){\line( 0,-1){180}}
}%
{\color[rgb]{0,0,0}\put(2251,-1321){\line( 0,-1){180}}
}%
{\color[rgb]{0,0,0}\put(181,-1771){\framebox(2340,270){}}
}%
{\color[rgb]{0,0,0}\multiput(451,-1501)(0.00000,-60.00000){5}{\line( 0,-1){ 30.000}}
}%
{\color[rgb]{0,0,0}\multiput(2251,-1501)(0.00000,-60.00000){5}{\line( 0,-1){ 30.000}}
}%
{\color[rgb]{0,0,0}\put(451,-1771){\line( 0,-1){ 90}}
}%
{\color[rgb]{0,0,0}\put(2251,-1771){\line( 0,-1){ 90}}
}%
{\color[rgb]{0,0,0}\put(181,-511){\line( 1, 0){1080}}
}%
{\color[rgb]{0,0,0}\put(181,-1051){\line( 1, 0){1080}}
}%
{\color[rgb]{0,0,0}\put(181,-1321){\framebox(1080,1080){}}
}%
{\color[rgb]{0,0,0}\multiput(181,-781)(58.37838,0.00000){19}{\line( 1, 0){ 29.189}}
}%
{\color[rgb]{0,0,0}\put(721,-511){\line( 0,-1){540}}
}%
{\color[rgb]{0,0,0}\multiput(991,-511)(0.00000,-60.00000){5}{\line( 0,-1){ 30.000}}
}%
{\color[rgb]{0,0,0}\multiput(451,-511)(0.00000,-60.00000){5}{\line( 0,-1){ 30.000}}
}%
\put(1981,-421){\makebox(0,0)[b]{\smash{{\SetFigFont{10}{12.0}{\familydefault}{\mddefault}{\updefault}{\color[rgb]{0,0,0}$h'$}%
}}}}
\put(1576,-691){\makebox(0,0)[b]{\smash{{\SetFigFont{10}{12.0}{\familydefault}{\mddefault}{\updefault}{\color[rgb]{0,0,0}$d_1'$}%
}}}}
\put(1846,-691){\makebox(0,0)[b]{\smash{{\SetFigFont{10}{12.0}{\familydefault}{\mddefault}{\updefault}{\color[rgb]{0,0,0}$d_2'$}%
}}}}
\put(2116,-691){\makebox(0,0)[b]{\smash{{\SetFigFont{10}{12.0}{\familydefault}{\mddefault}{\updefault}{\color[rgb]{0,0,0}$d_3'$}%
}}}}
\put(2386,-691){\makebox(0,0)[b]{\smash{{\SetFigFont{10}{12.0}{\familydefault}{\mddefault}{\updefault}{\color[rgb]{0,0,0}$d_4'$}%
}}}}
\put(1711,-961){\makebox(0,0)[b]{\smash{{\SetFigFont{10}{12.0}{\familydefault}{\mddefault}{\updefault}{\color[rgb]{0,0,0}$t_1'$}%
}}}}
\put(2251,-961){\makebox(0,0)[b]{\smash{{\SetFigFont{10}{12.0}{\familydefault}{\mddefault}{\updefault}{\color[rgb]{0,0,0}$t_2'$}%
}}}}
\put(1351,-1681){\makebox(0,0)[b]{\smash{{\SetFigFont{10}{12.0}{\familydefault}{\mddefault}{\updefault}{\color[rgb]{0,0,0}$B$}%
}}}}
\put(1981,-1231){\makebox(0,0)[b]{\smash{{\SetFigFont{10}{12.0}{\familydefault}{\mddefault}{\updefault}{\color[rgb]{0,0,0}$A'$}%
}}}}
\put(316,-691){\makebox(0,0)[b]{\smash{{\SetFigFont{10}{12.0}{\familydefault}{\mddefault}{\updefault}{\color[rgb]{0,0,0}$d_1$}%
}}}}
\put(586,-691){\makebox(0,0)[b]{\smash{{\SetFigFont{10}{12.0}{\familydefault}{\mddefault}{\updefault}{\color[rgb]{0,0,0}$d_2$}%
}}}}
\put(856,-691){\makebox(0,0)[b]{\smash{{\SetFigFont{10}{12.0}{\familydefault}{\mddefault}{\updefault}{\color[rgb]{0,0,0}$d_3$}%
}}}}
\put(1126,-691){\makebox(0,0)[b]{\smash{{\SetFigFont{10}{12.0}{\familydefault}{\mddefault}{\updefault}{\color[rgb]{0,0,0}$d_4$}%
}}}}
\put(721,-421){\makebox(0,0)[b]{\smash{{\SetFigFont{10}{12.0}{\familydefault}{\mddefault}{\updefault}{\color[rgb]{0,0,0}$h$}%
}}}}
\put(451,-961){\makebox(0,0)[b]{\smash{{\SetFigFont{10}{12.0}{\familydefault}{\mddefault}{\updefault}{\color[rgb]{0,0,0}$t_1$}%
}}}}
\put(991,-961){\makebox(0,0)[b]{\smash{{\SetFigFont{10}{12.0}{\familydefault}{\mddefault}{\updefault}{\color[rgb]{0,0,0}$t_2$}%
}}}}
\put(721,-1231){\makebox(0,0)[b]{\smash{{\SetFigFont{10}{12.0}{\familydefault}{\mddefault}{\updefault}{\color[rgb]{0,0,0}$A$}%
}}}}
\end{picture}%

%% file: netop.tex
\begin{figure} \label{fig:netop}
\begin{framed}

\begin{prooftree}
\AxiomC{$e \enginestep E \conn {} e'$}
\UnaryInfC{$
  \netconf {e \hastype E \cup \many {e \hastype E}} {\many m}
    \netstep{}
  \netconf {e' \hastype E \cup \many {e \hastype E}} {\many m}
$}
\end{prooftree}

\begin{prooftree}
\AxiomC{$e \enginestep E \conn {m} e'$}
\UnaryInfC{$
  \netconf {e \hastype E \cup \many {e \hastype E}} {\many m}
    \netstep{}
  \netconf {e' \hastype E \cup \many {e \hastype E}} {\{m\} \uplus \many m}
  $}
\end{prooftree}

\begin{prooftree}
\AxiomC{$e \enginestep E \conn {m^\bullet} e'$}
\UnaryInfC{$
\netconf {e \hastype E \cup \many {e \hastype E}} {\{m\} \uplus \many m}
    \netstep{}
  \netconf {e' \hastype E \cup \many {e \hastype E}} {\many m}
  $}
\end{prooftree}

\begin{prooftree}
\AxiomC{$\port \lP a \in A$}
\UnaryInfC{$
  \netconf {\many {e \hastype E}} {\{\rmsg a{\many d}\} \uplus \many m}
    \netstep{\rmsg a{\many d}}
  \netconf {\many {e \hastype E}} {\many m}
  $}
\end{prooftree}

\begin{prooftree}
\AxiomC{$\port \lO a \in A$}
\UnaryInfC{$
  \netconf {\many {e \hastype E}} {\many m}
    \netstep{{\rmsg a{\many d}}^\bullet}
  \netconf {\many {e \hastype E}} {\{\rmsg {\conn(a)}{\many d}\} \uplus \many m}
$}
\end{prooftree}

\end{framed}
\caption{Operational semantics of HRAM nets}
\end{figure}

%% file: category.tex
The category, called \netcat, is defined as follows:

\begin{itemize}
  \item Objects are interfaces $A \in \psetfin(\ports)$ identified up to
    \portnames-equivalence.
  \item A morphism $f \hastype A \tfun B$ is a well-formed net on the form \net
    {\many E}\conn{A \fun B}, for some $\many E$ and $\conn$.
    We will identify morphisms that have the same denotation, i.e. if $\denot f
    \pneq \denot g$ then $f = g$ (in the category).
  \item The identity morphism for an object $A$ is
    \[\id A \defeq \net \emptyset \conn {A \fun A'}\]
    for an $A'$ such that $\pi \vdash A \pneq A'$
    and
    \begin{align*}
      \conn(a) & \defeq \pi(a) \text{ if } a \in \supp({A \dual} \at O)
      \\ \conn(a) & \defeq \pi^{-1}(a) \text{ if } a \in \supp({A'} \at O) \text{.}
    \end{align*}
    Note that $A \fun A' = A \dual \cup A'$.
    This means that the identity is pure connectivity.
  \item Composition of two morphisms $f = \net {\many E_f}{\conn_f}{A \fun B}
  \hastype A \tfun B$ and $g = \net {\many E_g}{\conn_g}{B' \fun C} \hastype B'
  \tfun C$, such that $\pi \vdash B \pneq B'$, is
  \[f \compose g = \net {\many E_f \cup \many E_g}{\conn_{f\compose g}}{A \fun C}
      \hastype A \tfun C
    \]
  where
    \begin{align*}
      \conn_{f\compose g}(a) & \defeq \conn_f(a)  \text{ if }
          a \in \supp({A\dual} \at O \tensor I_f \at P)
          \wedge \conn_f(a) \notin \supp(B)
     \\ \conn_{f\compose g}(a) & \defeq \conn_g(a)  \text{ if }
          a \in \supp(C \at O \tensor I_g \at P)
          \wedge \conn_g(a) \notin \supp(B')
     \\ \conn_{f\compose g}(a) & \defeq \conn_g(\pi(\conn_f(a)))  \text{ if }
          a \in \supp({A\dual} \at O \tensor I_f \at P)
          \wedge \conn_f(a) \in \supp(B)
     \\ \conn_{f\compose g}(a) & \defeq \conn_f(\pi^{-1}(\conn_g(a)))  \text{ if }
          a \in \supp(C \at O \tensor I_g \at P)
          \wedge \conn_g(a) \in \supp(B')
    \end{align*}
  and
    \begin{align*}
         I_f & \defeq \tensor \{ A \mid \engine A P \in \many E_f \}
      \\ I_g & \defeq \tensor \{ A \mid \engine A P \in \many E_g \}
        \text{.}
    \end{align*}
\end{itemize}

\paragraph*{Note}
  We identify HRAMs with interfaces of the same shape in the category, which
  means that our objects and morphisms are in reality unions of equivariant
  sets.  In defining the operations of our category we use
  \emph{representatives} for these sets, and require that the representatives
  are chosen such that their sets of port names are disjoint (but same-shaped
  when the operation calls for it). The composition operation may appear to be
  partial because of this requirement, but we can always find equivariant
  representatives that fulfil it.

  It is possible to find other representations of interfaces that do not rely
  on equivariance. For instance, an interface could simply be two natural
  numbers --- the number of input and output ports.  Another possibility would
  be to make the tensor the \emph{disjoint} union operator. Both of these
  would, however, lead to a lot of bureaucracy relating to injection functions
  to make sure that port connections are routed correctly.  Our formulation,
  while seemingly complex, leads to very little bureaucracy, and is easy to
  implement.

\begin{proposition} \label{prop:netcat}
  \netcat is a category.
\end{proposition}
\begin{proof}
  \begin{itemize}
    \item Composition is well-defined, i.e. it preserves well-formedness.

      Let $f = \net {\many E_f}{\conn_f}{A \fun B} \hastype A \tfun B$ and
          $g = \net {\many E_g}{\conn_g}{B' \fun C} \hastype B' \tfun C$
      be morphisms such that $\pi \vdash B \pneq B'$, and their composition $f
      \compose g = \net {\many E_f \cup \many E_g}\conn{A \fun C} \hastype A
      \tfun C$ be as in the definition of composition. To prove that this is
      well-formed, we need to show that
      \begin{align*}
        \conn \in \supp((A \fun C) \at O \tensor I_{fg}\at P) & \tfun
        \supp((A \fun C) \at P \tensor I_{fg}\at O) =
        \\ \supp({A \dual} \at O \tensor C \at O \tensor I_f \at P \tensor I_g \at P) & \tfun
           \supp({A \dual} \at P \tensor C \at O \tensor I_f \at O \tensor I_g \at O)
      \end{align*}
      where $I_{fg} = \tensor\{A \mid \engine A P \in \many E_f
      \cup \many E_g\}$, and that it is a bijection.

      We are given that
      \begin{align*}
        \conn_f & \in \supp({A \dual} \at O  \tensor B \at O \tensor I_f \at P)
                \tfun \supp({A \dual} \at P \tensor B \at P \tensor I_f \at O)
        \\ \conn_g & \in \supp({{B'} \dual} \at O \tensor C \at O \tensor I_g \at P)
                   \tfun \supp({{B'} \dual} \at P \tensor C \at P \tensor I_g \at O)
        \\ \pi & \in \supp(B) \tfun \supp(B')
      \end{align*}
      are bijections.

      It is relatively easy to see that the domains specified in the clauses of
      the definition  of $\conn$ are mutually disjoint sets and that their
      union is the domain that we are after.

      Since $\conn$ is defined in clauses each of which defined
      using either $\conn_f$ or $\conn_g$ and/or $\pi$ (which are bijections with
      disjoint domains and codomains), it is enough to show that the set of
      port names that $\conn_f$ is applied to in clause 1 and 4 are disjoint, and
      similarly for $\conn_g$ in clause 2 and 3:

      \begin{itemize}
        \item In clause 4, we have $\conn_g(a) \in \supp(B')$, and so
          $\pi^{-1}(\conn_g(a)) \in \supp(B)$, which is disjoint from
          $\supp({A\dual} \at O \tensor I_f \at P)$ in clause 1.
        \item In clause 3, we have $\conn_f(a) \in \supp(B)$, and so
          $\pi(\conn_f(a)) \in \supp(B')$, which is disjoint from
          $\supp(C \at O \tensor I_g \at P)$ in clause 2.
      \end{itemize}

    \item Composition is associative.

      Let
      \begin{align*}
           f & = \net {\many E_f}{\conn_f}{A \fun B} \hastype A \tfun B \text{,}
        \\ g & = \net {\many E_g}{\conn_g}{B' \fun C} \hastype B' \tfun C \text{, and}
        \\ h & = \net {\many E_h}{\conn_h}{C' \fun D} \hastype C' \tfun D
      \end{align*}
          be nets such that $\pi_1 \vdash B \pneq B'$ and $\pi_2 \vdash C \pneq C'$.
          Then we have:
      \[
             (f \compose g) \compose h
         = \net {\many E_f \cup \many E_g \cup \many E_h}
                {\conn_{(f \compose g) \compose h}}{A \fun D}
      \]
      and
      \[
             f \compose (g \compose h)
         = \net {\many E_f \cup \many E_g \cup \many E_h}
                {\conn_{f \compose (g \compose h)}}{A \fun D}
      \]
      according to the definition of composition.
      We need to show that $\conn_{(f \compose g) \compose h} = \conn_{f
      \compose (g \compose h)}$, which implies that $(f \compose g) \compose h
      = f \compose (g \compose h)$.

    We do this by expanding the definitions, simplified using the following
    auxiliary function:
    \newcommand{\connect}[2]{\mathname{connect}(#1,#2)}
    \begin{align*}
         \connect c {A}(a) & \defeq a & \text{ if } a \notin \supp(A)
      \\ \connect c {A}(a) & \defeq c(a) & \text{ if } a \in \supp(A)
    \end{align*}

    $f \compose g = \net {\many E_f \cup \many E_g}{\conn_{f\compose g}}{A \fun C}$
    and $g \compose h = \net {\many E_g \cup \many E_h}{\conn_{g\compose h}}{B' \fun D}$
      where
    \begin{align*}
         \conn_{f\compose g}(a) & \defeq
           \connect {\conn_g \circ \pi_1}B(\conn_f(a))
            \text{ if } a \in \supp({A\dual} \at O \tensor I_f \at P)
      \\ \conn_{f\compose g}(a) & \defeq
           \connect {\conn_f \circ \pi_1^{-1}}{B'}(\conn_g(a))
            \text{ if } a \in \supp(C \at O \tensor I_g \at P)
      \\ \conn_{g\compose h}(a) & \defeq
           \connect {\conn_h \circ \pi_2}C(\conn_g(a))
            \text{ if } a \in \supp({{B'}\dual} \at O \tensor I_g \at P)
      \\ \conn_{g\compose h}(a) & \defeq
           \connect {\conn_g \circ \pi_2^{-1}}{C'}(\conn_h(a))
            \text{ if } a \in \supp(D \at O \tensor I_h \at P)
    \end{align*}

    Now $\conn_{(f \compose g) \compose h}$ and $\conn_{f \compose (g \compose h)}$ are defined as follows:
    \begin{align*}
         \conn_{(f\compose g) \compose h}(a) & \defeq
            \connect {\conn_h \circ \pi_2}C(\conn_{f \compose g}(a))
             \text{ if } a \in \supp({A\dual} \at O \tensor I_{f \compose g} \at P)
      \\ \conn_{(f\compose g) \compose h}(a) & \defeq
            \connect {\conn_{f \compose g} \circ \pi_2^{-1}}{C'}(\conn_h(a))
             \text{ if } a \in \supp(D \at O \tensor I_h \at P)
      \\ \conn_{f\compose (g \compose h)}(a) & \defeq
            \connect {\conn_{g \compose h} \circ \pi_1}B(\conn_f(a))
             \text{ if } a \in \supp({A\dual} \at O \tensor I_f \at P)
      \\ \conn_{f\compose (g \compose h)}(a) & \defeq
            \connect {\conn_f \circ \pi_1^{-1}}{B'}(\conn_{g \compose h}(a))
             \text{ if } a \in \supp(D \at O \tensor I_{g \compose h} \at P)
    \end{align*}

    One way to see that these two bijective functions are equal is to view them
    as case trees, and consider every case. There are 13 such cases to
    consider, out of which three are not possible.

    We show three cases:
    \begin{enumerate}
      \item If $a \in \supp({A \dual} \at O \tensor I_f \at P)$,
        $\conn_f(a) \notin \supp(B)$, and $\conn_f(a) \notin \supp(C)$, then
        \begin{align*}
          & \conn_{(f \compose g) \compose h}(a)
            \\ = & \connect {\conn_h \circ \pi_2}C(\conn_{f \compose g}(a))
            \\ = & \connect {\conn_h \circ \pi_2}C(\conn_f(a))
            \\ = & \conn_f(a)
        \end{align*}
        and
        \begin{align*}
          & \conn_{f \compose (g \compose h)}(a)
            \\ = & \connect {\conn_{g \compose h} \circ \pi_1}B(\conn_f(a))
            \\ = & \conn_f(a)
        \end{align*}
        and thus equal.
      \item Consider the case where $a \in \supp({A \dual} \at O \tensor I_f
        \at P)$, $\conn_f(a) \notin \supp(B)$, and $\conn_f(a) \in \supp(C)$.
        This case is not possible, since $\supp(C)$ is not a subset of the
        codomain of $\conn_f(a)$, which is $\supp({A \dual} \at P \tensor B \at
        P \tensor I_f \at O)$.
      \item If $a \in \supp(D \at O \tensor I_h \at P)$, $\conn_h(a) \in
        \supp(C')$, $\pi_2^{-1}(\conn_h(a)) \in \supp(C \at O \tensor I_g \at
        P)$, and $\conn_g(\pi_2^{-1}(\conn_h(a))) \in \supp(B')$, then
        \begin{align*}
          & \conn_{(f \compose g) \compose h}(a)
            \\ = & \connect {\conn_{f \compose g} \circ \pi_2^{-1}}{C'}(\conn_h(a))
            \\ = & \conn_{f \compose g}(\pi_2^{-1}(\conn_h(a)))
            \\ = & \connect {\conn_f \circ \pi_1^{-1}}{B'}(\conn_g(\pi_2^{-1}(\conn_h(a))))
            \\ = & \conn_f(\pi_1^{-1}(\conn_g(\pi_2^{-1}(\conn_h(a)))))
        \end{align*}
        and
        \begin{align*}
          & \conn_{f \compose (g \compose h)}(a)
          \\ = & \connect {\conn_f \circ \pi_1^{-1}}{B'}(\conn_{g \compose h}(a))
          \\ = & \connect {\conn_f \circ \pi_1^{-1}}{B'}(\connect {\conn_g \circ \pi_2^{-1}}{C'}(\conn_h(a)))
          \\ = & \connect {\conn_f \circ \pi_1^{-1}}{B'}(\conn_g(\pi_2^{-1}(\conn_h(a))))
          \\ = & \conn_f(\pi_1^{-1}(\conn_g(\pi_2^{-1}(\conn_h(a)))))
        \end{align*}
        and thus equal.
    \end{enumerate}

    The other cases are done similarly.

    \item $\id A$ is well-formed.
      For any interface $A$,
      \[\id A \defeq \net \emptyset \conn {A \fun A'}\]
      for an $A'$ such that $\pi \vdash A \pneq A'$
      and
      \begin{align*}
        \conn(a) & \defeq \pi(a) \text{ if } a \in \supp({A \dual} \at O)
        \\ \conn(a) & \defeq \pi^{-1}(a) \text{ if } a \in \supp({A'} \at O \text{.})
      \end{align*}
      according to the definition.

      We need to show that $\conn$ is a bijection:
    \begin{align*}
    \conn & \in \supp((A \fun A') \at O)
          \tfun \supp((A \fun A') \at P)
      \\ & = \supp({A \dual} \at O \cup {A'} \at O)
       \tfun \supp({A \dual} \at P \cup {A'} \at P)
    \end{align*}
    This is true since $\pi$ is a bijection in $\supp(A) \tfun \supp(A')$.

    \item $\id A$ is an identity. For any morphism $f \hastype A \tfun B$
    we observe that $\id A \compose f$ is structurally equivalent to $f$, so by
    Theorem~\ref{thm:structisdenot}, $\denot {\id A \compose f} \pneq \denot
    f$.

    The case for $f \compose \id B$ is similar.
    \end{itemize}
\end{proof}

  We will now show that \netcat is a symmetric monoidal category:
  \begin{itemize}
    \item The tensor product of two objects $A,B$, $A \tensor B$ has already
      been defined.  We define the tensor of two morphisms
        $f = \net {\many E_f}{\conn_f}{A \fun B},
         g = \net {\many E_g}{\conn_g}{C \fun D}$
      as $f \tensor g = \net {\many E_f \cup \many E_g}{\conn_f \tensor \conn_g}{A
      \tensor C \fun B \tensor D}$.
    \item The unit object is the empty interface, $\emptyset$.
    \item Since $A \tensor (B \tensor C) = A \cup B \cup C = (A \tensor B) \tensor C$ we
      define the associator $\assoc A B C \defeq \id {A \tensor B \tensor C}$ with the
      obvious inverse.
    \item Similarly, since $\emptyset \tensor A = \emptyset \cup A = A = A \cup \emptyset = A \tensor \emptyset$, we define the left unitor $\lunitor A \defeq \id A$ and the right unitor
    $\runitor A \defeq \id A$.
    \item Since $A \tensor B = A \cup B = B \cup A = B \tensor A$ we define the
      commutativity constraint $\commut A B \defeq \id {A \tensor B}$.
  \end{itemize}
  \begin{proposition}
    \netcat is a symmetric monoidal category.
  \end{proposition}
  \begin{proof}
    \begin{itemize}
      \item The tensor product is well-defined, i.e. for two morphisms $f,g$,
        $f \tensor g$ is a well-formed net. This is easy to see since $f$ and
        $g$ are well-formed.
      \item The tensor product is a bifunctor:
        \begin{itemize}
          \item $\id A \tensor \id B = \net \emptyset {\conn_1 \tensor \conn_2} {A \tensor
            B \fun A' \tensor B'} = \id {A \tensor B}$ by the definition of $\id {A \tensor B}$.
          \item $(f \compose g) \tensor (h \compose i) = f \tensor h \compose g \tensor i$
          by the definition of composition and tensor on morphisms.
        \end{itemize}
      \item The coherence conditions of the natural isomorphisms are trivial
        since the isomorphisms amount to identities.
    \end{itemize}
  \end{proof}

  Next we show that \netcat is a compact-closed category:
  \begin{itemize}
    \item We have already defined the dual $A \dual$ of an object $A$.
    \item Since $\emptyset \fun (A \dual \tensor A') = \emptyset \dual \cup (A
      \dual \cup A') = A \fun A'$ we can define the unit $\unit A \defeq \id A$
      and since $A \tensor {A'}\dual \fun \emptyset = (A \cup {A'}\dual)\dual \cup
      \emptyset = A \dual \cup A' = A \fun A'$ we can define the counit
      $\counit A \defeq \id A$.

  \end{itemize}

  This leads us directly to the following result --- what we set out to show:
  \begin{proposition}
    \netcat is a symmetric compact-closed category.
  \end{proposition}

The following two theorems can be proved by induction on the trace length,
and provide a connection between the \netcat tensor and composition and
trace interleaving and composition.
\begin{theorem} \label{thm:tensordenot}
  If $f \hastype A \tfun B$ and $g \hastype C \tfun D$ are morphisms of \netcat
  then $\denot{f \tensor g} = \denot{f} \tensor \denot{g}$.
\end{theorem}

\begin{theorem} \label{thm:compdenot}
  If $f \hastype A \tfun B$ and $g \hastype B' \tfun C$ are morphisms of
  \netcat such that $\pi \vdash B \pneq B'$ then $\denot{f \compose g} =
  \denot{f} \compose \denot{g}$.
\end{theorem}

%% file: combination.tex
\begin{proof}
We show that for any trace $s$, $s \in \denot{S}$ implies $s \in \denot{S'}$ by
induction on the length of the trace.
\begin{description}
  \item[Hypothesis.] If $s \in \denot{S}$ and thus
    $\initial(S) \netstep{s} \netconf {\{\engineconf {\many t_1}{h_1} \hastype
    E_1,\engineconf {\many t_2}{h_2} \hastype E_2\}} {\many m}$ for some sets of
    threads $\many t_1$ and $\many t_2$, heaps $h_1$ and $h_2$, and a multiset of
    messages $\many m$, then $\initial(S') \netstep{s} \netconf {\{\engineconf
    {\many t_1 \cup \many t_2 \cup \many t_p}{h_1 \cup h_2} \hastype E_{12}\}} {\many m_p}$
    where $\many t_p$ is a set of threads and $\many m_p$ is a multiset of messages such that:
    \begin{enumerate}
      \item each $t \in \many t_p$ is on the form
      $t = \thread {\nispark\ a}{\many d}$ with
      $\conn(a) \in \supp(A_1 \tensor A_2)$, and
      \item $\many m = \many m_p \uplus \{ \rmsg {\conn(a)}{\msgf(\many d)}
          \mid \thread {\nispark\ a}{\many d} \in \many t_p \}$.
    \end{enumerate}

    Intuitively, the net where $E_1$ and $E_2$ have been
    combined into one engine will not have pending messages (in $\many m$) for
    communications between $E_1$ and $E_2$, but it can match the behaviour of such
    messages by threads that are just about to spark.
  \item[Base case.] Since any net can take zero steps, the case when $s = \emptytrace$
    is trivial.
  \item[Inductive step.] If $s = s' \snoc \alpha$ and the hypothesis holds for $s'$, then
    we have
    \begin{align*}
      \initial(S) & \netstep{s'} \netconf{\{\engineconf {\many t_1}{h_1} \hastype E_1
                                         ,\engineconf {\many t_2}{h_2} \hastype E_2\}}
                                       {\many m}
      \\          & \netstep{}^* \netstep{\alpha}
                    \netconf {\{\engineconf {\many t_1'}{h_1'} \hastype E_1
                                         ,\engineconf {\many t_2'}{h_2'} \hastype E_2\}}
                                       {\many m'}
      \\ \initial(S') & \netstep{s'} \netconf{\{\engineconf {\many t_1 \cup \many
           t_2 \cup \many t_p}{h_1 \cup h_2} \hastype E_{12}\}} {\many m_p}
    \end{align*}
    with $t_p$ and $\many m'$ as in the hypothesis.
    We first show that $S'$ can match the silent steps that $S$ performs, by induction
    on the number of steps, using the same induction hypothesis as above:
    \begin{description}
      \item[Base case.] Trivial.
      \item[Inductive step.]
      Assume that we have
      \begin{align*}
        \initial(S) & \netstep{s'}\netstep{}^*
          \netconf{\{\engineconf {\many t_1}{h_1} \hastype E_1
                    ,\engineconf {\many t_2}{h_2} \hastype E_2\}}
                  {\many m}
        \\ \initial(S') & \netstep{s'}\netstep{}^*
            \netconf{\{\engineconf {\many t_1 \cup \many t_2 \cup \many t_p}{h_1 \cup h_2}
                                    \hastype E_{12}\}}
                    {\many m_p}
      \end{align*}
        Such that the induction hypothesis holds.
        We need to show that any step
        \begin{align*}
          \netconf{\{\engineconf {\many t_1}{h_1} \hastype E_1
                    ,\engineconf {\many t_2}{h_2} \hastype E_2\}}
                  {\many m}
          & \netstep{} \\
           \netconf {\{\engineconf {\many t_1'}{h_1'} \hastype E_1
                      ,\engineconf {\many t_2'}{h_2'} \hastype E_2\}}
                    {\many m'}
        \end{align*}
        can be matched by (any number of) silent steps of the $S'$
        configuration, such that the induction hypothesis still holds.
        \begin{itemize}
          \item A thread of $S$ performs a silent step. This is trivial, since
            the threads of the engine configuration of $S'$ includes all
            threads of the configurations of $S$, and its heap is the union of
            those of $S$.  \item A thread of $S$ does an internal engine send
            step. Since $\many t_1 \cup \many t_2 \cup \many t_p$ includes all
            threads of the $S$ configuration, and for the port name $a$ in
            question $\conn(a) \in A_1 \cup A_2 = A_1 \tensor A_2$, this can be
            matched by the configuration of $S'$ such that the induction
            hypothesis still holds.

          \item A thread $S$ does an external engine send. This means that there
            is a thread $t \in \many t_1 \cup \many t_2$ on the form $t =
            \thread {\nispark\ a}{\many d}$, which after the step will be
            removed, adding the message $\rmsg {\conn(a)}{\msgf(\many d)}$ to its multiset of
            messages, i.e. $\many m' = \many m \uplus \{\rmsg {\conn(a)}{\msgf(\many d)}\}$.

            If $\conn(a) \in A_1 \cup A_2$, then the configuration $S'$ can
            take zero steps, and thus include $t$ in the set of threads ready
            to spark. The induction hypothesis still holds, since $\many m' =
            \many m \uplus \{\rmsg {\conn(a)}{\msgf(\many d)}\} = \many m_p \uplus 
            \{\rmsg {\conn(a)}{\msgf(\many d)} \mid \thread {\nispark\ a}{\many d} \in \many
            t_p\} \uplus \{\rmsg {\conn(a)}{\msgf(\many d)}\} = \many m_p \uplus \{\rmsg
            {\conn(a)}{\msgf(\many d)} \mid \thread {\nispark\ a}{\many d} \in \many t_p
            \cup \{t\}\}$.

            If $\conn(a) \in I$, then the configuration of $S'$ can match the
            step of $S$, removing the thread $t$ from also its set of threads.
            It is easy to see that the induction hypothesis holds also in this
            case.
          \item An engine of $S$ receives a message. This means that
            $\many m = \{\rmsg a{\many d}\} \uplus \many m'$ for a message such
            that the port $\port \lO a \in A_1 \cup A_2 = A_1 \tensor A_2$.
            Then either $\rmsg a{\many d}$ is in $\many m_p$ or in $\{ \rmsg
            {\conn(a)}{\msgf(\many d)} \mid \thread {\nispark\ a}{\many d} \in
            \many t_p \}$. If it is the former, $E_{12}$ can receive the
            message and start a thread equal to that started in the
            configuration of $S$.  If it is the latter, there is a thread $t =
            \thread {\nispark\ {\conn^{-1}(a)}}{\many d'} \in \many t_p$ with
            $\many d = \msgf(\many d')$ that can first take a send $m$ step,
            adding it to the multiset of pending messages of the configuration
            of $S'$, and then it can be received as in $S$.
        \end{itemize}
    \end{description}
    Next we show that the $\alpha$ step can be matched:
      Assume that we have
      \begin{align*}
        \initial(S) & \netstep{s'}\netstep{}^*
          \netconf{\{\engineconf {\many t_1}{h_1} \hastype E_1
                    ,\engineconf {\many t_2}{h_2} \hastype E_2\}}
                  {\many m}
        \\ \initial(S') & \netstep{s'}\netstep{}^*
            \netconf{\{\engineconf {\many t_1 \cup \many t_2 \cup \many t_p}{h_1 \cup h_2}
                                    \hastype E_{12}\}}
                    {\many m_p}
      \end{align*}
        Such that the induction hypothesis holds.
        We need to show that for any $\alpha$, a step
        \begin{align*}
          \netconf{\{\engineconf {\many t_1}{h_1} \hastype E_1
                    ,\engineconf {\many t_2}{h_2} \hastype E_2\}}
                  {\many m}
          & \netstep{\alpha}
           \netconf {\{\engineconf {\many t_1'}{h_1'} \hastype E_1
                      ,\engineconf {\many t_2'}{h_2'} \hastype E_2\}}
                    {\many m'}
        \end{align*}
        can be matched by the $S'$ configuration, such that the induction
        hypothesis still holds.  We have two cases:
        \begin{itemize}
          \item The configuration of $S$ performs a send step. That is
            $\many m = \{m\} \uplus \many m'$ for an $m = \rmsg a{\many d}$ such that
            $\port \lP a \in A$. Since $\supp(A)$ is disjoint
            from $\supp(A_1 \cup A_2)$, the message is also in $\many m_p$, so
            the configuration of $S'$ can match the step.
          \item The configuration of $S$ performs a receive step. This case
            is easy, as $S$ and $S'$ have the same interface $A$.
        \end{itemize}

\end{description}

\end{proof}

%% file: nonlocal.latex
\setlength{\unitlength}{4144sp}%
\begingroup\makeatletter\ifx\SetFigFont\undefined%
\gdef\SetFigFont#1#2#3#4#5{%
  \reset@font\fontsize{#1}{#2pt}%
  \fontfamily{#3}\fontseries{#4}\fontshape{#5}%
  \selectfont}%
\fi\endgroup%
\begin{picture}(3084,1194)(214,-1153)
\thinlines
{\color[rgb]{0,0,0}\put(1531,-1141){\framebox(450,450){}}
}%
{\color[rgb]{0,0,0}\put(2611,-916){\vector(-1, 0){630}}
}%
{\color[rgb]{0,0,0}\put(451,-916){\line(-1, 0){225}}
\put(226,-916){\line( 0, 1){720}}
\put(226,-196){\vector( 1, 0){2385}}
}%
{\color[rgb]{0,0,0}\put(451,-1141){\framebox(450,450){}}
}%
{\color[rgb]{0,0,0}\put(2611,-1141){\framebox(450,1170){}}
}%
{\color[rgb]{0,0,0}\put(3286,-556){\vector(-1, 0){225}}
}%
{\color[rgb]{0,0,0}\put(1531,-916){\vector(-1, 0){630}}
}%
\put(1756,-961){\makebox(0,0)[b]{\smash{{\SetFigFont{10}{12.0}{\familydefault}{\mddefault}{\updefault}{\color[rgb]{0,0,0}$S$}%
}}}}
\put(1756,-106){\makebox(0,0)[b]{\smash{{\SetFigFont{10}{12.0}{\familydefault}{\mddefault}{\updefault}{\color[rgb]{0,0,0}$(p_2,p_3)$}%
}}}}
\put(676,-961){\makebox(0,0)[b]{\smash{{\SetFigFont{10}{12.0}{\familydefault}{\mddefault}{\updefault}{\color[rgb]{0,0,0}$S$}%
}}}}
\put(2836,-601){\makebox(0,0)[b]{\smash{{\SetFigFont{10}{12.0}{\familydefault}{\mddefault}{\updefault}{\color[rgb]{0,0,0}$C[-]$}%
}}}}
\put(2296,-826){\makebox(0,0)[b]{\smash{{\SetFigFont{10}{12.0}{\familydefault}{\mddefault}{\updefault}{\color[rgb]{0,0,0}$(p_0,p_1)$}%
}}}}
\put(1216,-826){\makebox(0,0)[b]{\smash{{\SetFigFont{10}{12.0}{\familydefault}{\mddefault}{\updefault}{\color[rgb]{0,0,0}$(p_1,p_2)$}%
}}}}
\end{picture}%

%% file: legality.tex
\begin{definition}
  The coabstracted and free pointers $\cp$ and $\fp \in \traces \tfun
  \pset(\pointers)$ are:
  \begin{align*}
       \cp(\emptytrace) & \defeq \emptyset
    \\ \cp(s\snoc\lmsg lap{p'}d) & \defeq \cp(s) \cup \{p'\}
    \\ \fp(\emptytrace) & \defeq \emptyset
    \\ \fp(s\snoc\lmsg lap{p'}d & \defeq \fp(s) \cup (\{p\} \setminus \cp(s))
  \end{align*}

\end{definition}
  The pointers of a trace $\ptrs(s) = \cp(s) \cup \fp(s)$.

\begin{definition}
  Define $\enabled_{\mf A} \in \traces_A \tfun \pset(\supp(A) \tprod \pointers)$ inductively as follows:
  \begin{align*}
       \enabled_{\mf A}(\emptytrace) & \defeq \emptyset
    \\ \enabled_{\mf A}(s\snoc \lmsg lap{p'}d) & \defeq \enabled_{\mf A}(s)
      \cup \{ \tuple{a',p'} \mid a \enables A a' \}
  \end{align*}
\end{definition}

\begin{definition}
  We define the following relations over traces:
  \begin{itemize}
    \item Write $s' \leq s$ if and only if there is a trace $s_1$ such that
      $s' \snoc s_1 = s$, i.e. $s'$ is a \emph{prefix} of $s$.
    \item Write $s' \leq s$ if and only if there are traces $s_1, s_2$ such that
      $s_1 \snoc s' \snoc s_2 = s$, i.e. $s'$ is a \emph{segment} of $s$.
  \end{itemize}
\end{definition}

\begin{definition}
  For an arena $\mf A$ and a trace $s \in \traces_A$, we define the following
  legality conditions:
  \begin{itemize}
    \item $s$ has \emph{unique pointers} when
      $s' \snoc \lmsg lap{p'}d \leq s$ implies $p' \notin \ptrs(s')$.

    \item $s$ is \emph{correctly labelled} when
      $\lmsg lap{p'}d \subseteq s$ implies $a \in \supp(A^{(l)})$.

    \item $s$ is \emph{justified} when $s' \snoc \lmsg l ap{p'}d \leq s$ and $a
      \notin \ini A$ implies $\tuple{a,p} \in \enabled_{\mf A}(s')$.

    \item $s$ is \emph{well-opened} when $s' \snoc \lmsg
      lap{p'}d\leq s$ implies $a \in \ini A$ and $s' = \emptytrace$.

    \item $s$ is \emph{strictly scoped} when $\lmsg l ap{p'}d \snoc s' \subseteq s$ with
    $a \in \ans A$ implies $p \notin \fp(s')$.

    \item $s$ is \emph{strictly nested} when
      $\lmsg {l_1}{a_1}p{p'}{d_1} \snoc s' \snoc \lmsg {l_2}{a_2}{p'}{p''}{d_2} \snoc
       \\ s'' \snoc \lmsg {l_3}{a_3}{p'}{p'''}{d_3}
       \subseteq s$ implies $\lmsg {l_4}{a_4}{p''}{-}{d_4} \subseteq s''$ for
       port names $a_1,a_2 \in \qst A$ and $a_3,a_4 \in \ans A$.

    \item $s$ is \emph{alternating} when $\tuple {l_1,m_1} \snoc \tuple{l_2,m_2}
      \subseteq s$ implies $l_1 \neq l_2$.
  \end{itemize}
\end{definition}

\begin{definition}
  We say that a question message $\alpha = \lmsg lap{p'}d$ ($a \in \qst A$) is
  \emph{pending} in a trace $s = s_1 \snoc \alpha \snoc s_2$ if and only if
  there is no answer $\alpha' = \lmsg {l'}{a'}{p'}{p''}{d'} \subseteq s_2$ ($a'
  \in \ans A$), i.e. the question has not been answered.
\end{definition}

Write $P_{\mf A}$ for the subset of $\traces_A$ consisting of the traces that
have unique pointers, are correctly labelled, justified, strictly scoped and
strictly nested.

For a set of traces $P$, write $P\cond{alt}$ for the subset consisting
of only alternating traces, and $P\cond{st}$ (for single-threaded) for
the subset consisting of only well-opened traces.

\begin{definition}
  If $s \in \traces$ and $X \subseteq \pointers$, define the \emph{hereditarily
  justified} trace $s \hereditary X$, where inductively $\tuple{s',X'} = s \hereditary X$:

  \begin{align*}
       \emptytrace \hereditary X & \defeq \tuple{\emptytrace, X}
    \\ s \snoc \lmsg lap{p'}d \hereditary X & \defeq \tuple{s'\snoc \lmsg lap{p'}d, B \cup \{p'\}}
      & \text{ if } p \in X'
    \\ s \snoc \lmsg lap{p'}d \hereditary X & \defeq \tuple{s',B}
      & \text{ if } p \notin X'
  \end{align*}
\end{definition}

Write $s \hereditary X$ for $s'$ when $s \hereditary X = \tuple{s',X'}$ when it is
convenient.

%% file: copycat.tex
\begin{lemma} \label{lem:extramessage}
  If $n_1 = \netconf {\many {e \hastype E}} {\many m}$
  and $n_1' = \netconf {\many {e' \hastype E}} {\many m'}$
  are net configurations of a net
  $f = \net {\many E} \conn A$, and
  $n_1 \netstep {(x)} n_1'$
    ($(x) \in \{\bullet\} \cup (\polarities \tprod \msgs_{\supp(A)}$)
  then
  $n_2 \netstep {(x)} n_2'$ where
  $n_2 = \netconf {\many {e \hastype E}} {\many m \uplus \{m\}}$ and
  $n_2' = \netconf {\many {e' \hastype E}} {\many m' \uplus \{m\}}$.
\end{lemma}
\begin{proof}
  By cases on $(x)$:
  \begin{itemize}
    \item
      If $(x) = \bullet$, then
      $\many {e \hastype E} = \{e \hastype E\} \cup \many {e_0 \hastype E_0}$,
      $e \enginestep E \conn {(y)} e'$ for some $(y)$,
      $\many {e' \hastype E} = \{e' \hastype E\} \cup \many {e_0' \hastype E_0}$.
      We have three cases for $(y)$:
      \begin{itemize}
        \item
          If $(y) = \bullet$, then
          $e \enginestep E \conn {} {e'}$ and
          $\many m' = \many m$. Then we also have
          $n_2 =
            \netconf {\{e \hastype E\} \cup \many {e_0 \hastype E_0}}
                     {\many m \uplus \{m\}}
            \netstep {}
            \netconf {\{e' \hastype E\} \cup \many {e_0' \hastype E_0}}
                     {\many m \uplus \{m\}}
          = n_2'$.
        \item
          If $(y) = \tuple{\lP,m'}$, then
          $e \enginestep E \conn {m'} {e'}$ and
          $\many m' = \{m'\} \cup \many m$. Then we also have
          $n_2 =
            \netconf {\{e \hastype E\} \cup \many {e_0 \hastype E_0}}
                     {\many m \uplus \{m\}}
            \netstep {}
            \netconf {\{e' \hastype E\} \cup \many {e_0' \hastype E_0}}
                     {\{m'\} \uplus \many m \uplus \{m\}}
          = n_2'$.
        \item
          If $(y) = \tuple{\lO,m'}$, then
          $e \enginestep E \conn {\overline{m'}} {e'}$ and
          $\many m = \{m'\} \uplus \many m'$. Then we also have
          $n_2 =
            \netconf {\{e \hastype E\} \cup \many {e_0 \hastype E_0}}
                     {\{m'\} \uplus \many m' \uplus \{m\}}
            \netstep {}
            \netconf {\{e' \hastype E\} \cup \many {e_0' \hastype E_0}}
                     {\many m' \uplus \{m\}}
          = n_2'$.
      \end{itemize}
    \item
      If $(x) = \tuple {\lP, m'}$, then
      $\many {e' \hastype E} = \many {e \hastype E}$ and
      $\many m = \{m'\} \uplus \many m'$. Then we also have
      $n_2 =
        \netconf {\many {e \hastype E}}
                 {\{m'\} \uplus \many m' \uplus \{m\}}
        \netstep {m'}
        \netconf {\many {e \hastype E}}
                 {\many m' \uplus \{m\}}
      = n_2'$.
    \item
      If $(x) = \tuple {\lO, m'}$, where $m' = \msg ap{p'}d$ then
      $\many {e' \hastype E} = \many {e \hastype E}$ and
      $\many m' = \{\msg {\conn(a)}p{p'}d\} \uplus \many m$. Then we also have
      $n_2 =
        \netconf {\many {e \hastype E}}
                 {\many m \uplus \{m\}}
        \netstep {\overline{m'}}
        \netconf {\many {e \hastype E}}
                 {\{\msg {\conn(a)}p{p'}d\} \uplus \many m \uplus \{m\}}
      = n_2'$.
  \end{itemize}
\end{proof}
\begin{lemma} \label{thm:swappings} If $f$ is a net and $s$ a trace, then
  \begin{enumerate}
    \item
      $s = s_1 \snoc \tuple{l,m_1} \snoc \tuple{\lO,m} \snoc s_2 \in \denot f$
      with witness $\initial(f) \netstep s n$ implies
      $s' = s_1 \snoc \tuple{\lO,m} \snoc \tuple{l,m_1} \snoc s_2 \in \denot f$
      with $\initial(f) \netstep {s'} n$ and
    \item
      $s = s_1 \snoc \tuple{\lP,m} \snoc \tuple{l,m_1} \snoc s_2 \in \denot f$
      with witness $\initial(f) \netstep s n$ implies
      $s' = s_1 \snoc \tuple{l,m_1} \snoc \tuple{\lP,m} \snoc s_2 \in \denot f$
      with $\initial(f) \netstep {s'} n$.
  \end{enumerate}
\end{lemma}
A special case of this theorem is that if $G = \gnet f {\mf A}$ and,
for a set of traces $S \subseteq P_{\mf A}$, $S \subseteq \denot G$ holds,
then $\sat(S) \subseteq \denot G$.
\begin{proof}
  \begin{enumerate}
    \item
      $s = s_1 \snoc \tuple{l,m_1} \snoc \tuple{\lO,m} \snoc s_2 \in \denot f$
      means that
      \[ \initial(f) \netstep {s_1} \netstep{(x)}^* n_1
                     \netstep {\tuple{l,m_1}} n_2 \netstep{(y)}^*
                     \netstep {\tuple{\lO,m}} n_3 \netstep{(z)}^*
                     \netstep {s_2} n_4
      \]
      for net configurations $n_1,n_2,n_3,n_4$. For clarity, we take
      $(x),(y),(z)$ to be ``names'' for the silent transitions.
      We show that there exist $n_2'$ and $(y')$ such that
      \[ \initial(f) \netstep {s_1} \netstep{(x)}^* n_1
                     \netstep {\tuple{\lO,m}}
                     \netstep {\tuple{l,m_1}} n_2' \netstep{(y')}^* n_3 \netstep{(z)}
                     \netstep {s_2} n
      \]
      by induction on the length of $\netstep{(y)}^*$:
      \begin{itemize}
        \item Base case.
          If $\netstep{(y)}^*$ is the identity relation, then assume
          \[ n_1 \netstep {\tuple{l,m_1}} n_2
                 \netstep {\tuple{\lO,m}} n_3
          \]
          Let $n_1 = \netconf {\many {e_1 \hastype E}} {\many m_1}$,
          $n_2 = \netconf {\many {e_2 \hastype E}} {\many m_2}$,
          $m = \msg ap{p'}d$, and
          $m' = \msg {\conn(a)}p{p'}d$.
          Then
          $n_3 = \netconf {\many {e_2 \hastype E}} {\{m'\} \uplus \many m_2}$ by
          the definition of $\netstep{}$.
          Since $\port \lO a \in I$,
          $n_1 \netstep {\tuple{\lO,m}}
            \netconf {\many {e_1 \hastype E}}
                     {\{m'\} \uplus \many m_1}$.
          Also, since $n_1 \netstep {\tuple{l,m_1}} n_2$
          we have
          $\netconf {\many {e_1 \hastype E}} {\{m'\} \uplus \many m_2}
            \netstep {\tuple{l,m_1}}
            n_3$
          by Lemma~\ref{lem:extramessage}.
          Composing the relations, we get
          \[ n_1 \netstep {\tuple{\lO,m}}
                 \netstep {\tuple{l,m_1}} n_3
          \]
          which completes the base case.
        \item Inductive step.
          If $\netstep{(y)}^* = \netstep{(y_0)}^*\netstep{\bullet}$ such that for
          any $n_3'$
          \[ n_1 \netstep {\tuple{l,m_1}} n_2
                 \netstep{(y_0)}^* \netstep {\tuple{\lO,m}} n_3'
          \]
          implies that there exist $n_2'$ and $(y_0')$ with
          \[ n_1 \netstep {\tuple{\lO,m}}
                 \netstep {\tuple{l,m_1}} n_2'
                 \netstep{(y_0')}^* n_3'
          \]
          then assume
          \[ n_1 \netstep{\tuple{l,m_1}} n_2
                 \netstep{(y_0)}^* n_{y_0}
                 \netstep{\bullet} n_y
                 \netstep {\tuple{\lO,m}} n_3
          \]
          Let $n_{y_0} = \netconf {\many {e_{y_0} \hastype E}} {\many m_{y_0}}$,
              $n_y     = \netconf {\many {e_y \hastype E}} {\many m_y}$,
          $m = \msg ap{p'}d$, and
          $m' = \msg {\conn(a)}p{p'}d$.
          Then $n_3 = \netconf {\many {e_y \hastype E}} {\{m'\} \uplus \many m_y}$
          by the definition of $\netstep{}$.
          Since $\port \lO a \in I$,
          $n_{y_0} \netstep {\tuple{\lO,m}}
            \netconf {\many {e_{y_0} \hastype E}}
                     {\{m'\} \uplus \many m_{y_0}}$.
          Also, since $n_{y_0} \netstep {\bullet} n_y$
          we have
          $\netconf {\many {e_{y_0} \hastype E}} {\{m'\} \uplus \many m_{y_0}}
            \netstep {\bullet}
            n_3$
          by Lemma~\ref{lem:extramessage}.
          Composing the relations, we get
          \[ n_1 \netstep {\tuple{l,m_1}} n_2
                 \netstep {(y_0)}^* n_{y_0}
                 \netstep {\tuple{\lO,m}} \netconf {\many {e_{y_0} \hastype E}}
                                                   {\{m'\} \uplus \many m_{y_0}}
                 \netstep {\bullet}
                 n_3
          \]
          Applying the hypothesis, we finally get
          \[ n_1 \netstep {\tuple{\lO,m}}
                 \netstep {\tuple{l,m_1}} n_2'
                 \netstep{(y_0')}^*
                 \netstep{\bullet} n_3
          \]
          which completes the first part of the proof.
      \end{itemize}
    \item
      $s = s_1 \snoc \tuple{\lP,m} \snoc \tuple{l,m_1} \snoc s_2 \in \denot f$
      means that
      \[ \initial(f) \netstep {s_1} \netstep{(x)}^* n_1
                     \netstep {\tuple{\lP,m}} n_2 \netstep{(y)}^*
                     \netstep {\tuple{l,m_1}} n_3 \netstep{(z)}^*
                     \netstep {s_2} n_4
      \]
      for net configurations $n_1,n_2,n_3,n_4$ and $(x),(y),(z)$ names for the
      silent transitions.
      We show that there exist $(y')$ and $n_2'$ such that
      \[ \initial(f) \netstep {s_1} \netstep{(x)}^* n_1
                     \netstep {(y')}^* n_2'
                     \netstep {\tuple{l,m_1}}
                     \netstep {\tuple{\lP, m}} n_3
                     \netstep{(z)}^*
                     \netstep {s_2} n
      \]
      by induction on the length of $\netstep{(y)}^*$:
      \begin{itemize}
        \item Base case.
          If $\netstep{(y)}^*$ is the identity relation, then assume
          \[ n_1 \netstep {\tuple{\lP,m}} n_2
                 \netstep {\tuple{l,m_1}} n_3
          \]
          Let $n_2 = \netconf {\many {e_2 \hastype E}} {\many m_2}$,
          $n_3 = \netconf {\many {e_3 \hastype E}} {\many m_3}$,
          $m = \msg ap{p'}d$
          Then
          $n_1 = \netconf {\many {e_2 \hastype E}} {\{m\} \uplus \many m_2}$ by
          the definition of $\netstep{}$.
          Since $\port \lP a \in I$,
          $ \netconf {\many {e_3 \hastype E}} {\{m\} \uplus \many m_3}
            \netstep {\tuple{\lP,m}} n_3$.
          Also, since $n_2 \netstep {\tuple{l,m_1}} n_3$
          we have
          $n_1
            \netstep {\tuple{l,m_1}}
            \netconf {\many {e_3 \hastype E}} {\{m\} \uplus \many m_3}$
          by Lemma~\ref{lem:extramessage}.
          Composing the relations, we get
          \[ n_1 \netstep {\tuple{l,m_1}}
                 \netstep {\tuple{\lP,m}} n_3
          \]
          which completes the base case.
        \item Inductive step.
          If $\netstep{(y)}^* = \netstep{\bullet}\netstep{(y_0)}^*$ such that for
          any $n_1'$
          \[ n_1' \netstep {\tuple{\lP,m}}
                 \netstep{(y_0)}^* n_2
                 \netstep {\tuple{l,m_1}} n_3
          \]
          implies that there exist $n_2'$ and $(y_0')$ with
          \[ n_1' \netstep{(y_0')}^* n_2'
                  \netstep {\tuple{l,m_1}}
                  \netstep {\tuple{\lP,m}} n_3
          \]
          then assume
          \[ n_1 \netstep{\tuple{\lP,m}} n_m
                 \netstep{\bullet} n_y
                 \netstep{(y_0)}^* n_2
                 \netstep {\tuple{l,m_1}} n_3
          \]
          Let $n_m     = \netconf {\many {e_m \hastype E}} {\many m_m}$,
              $n_y     = \netconf {\many {e_y \hastype E}} {\many m_y}$, and
          $m = \msg ap{p'}d$.
          Then $n_1 = \netconf {\many {e_m \hastype E}} {\{m\} \uplus \many m_m}$
          by the definition of $\netstep{}$.
          Since $\port \lP a \in I$,
          $\netconf {\many {e_y \hastype E}}
                    {\{m\} \uplus \many m_y}
            \netstep {\tuple{\lP,m}} n_y$.
          Also, since $n_m \netstep {\bullet} n_y$
          we have
          $n_1
            \netstep {\bullet}
            \netconf {\many {e_y \hastype E}}
                     {\{m\} \uplus \many m_y}$
          by Lemma~\ref{lem:extramessage}.
          Composing the relations, we get
          \[ n_1 \netstep {\bullet} \netconf {\many {e_y \hastype E}}
                                             {\{m\} \uplus \many m_y}
                 \netstep {\tuple{\lP,m}} n_y
                 \netstep{(y_0)}^* n_2
                 \netstep {\tuple{l,m_1}} n_3
          \]
          Applying the hypothesis, we finally get
          \[ n_1 \netstep{\bullet}
                 \netstep{(y_0')}^* n_2'
                 \netstep{\tuple{l,m_1}}
                 \netstep{\tuple{\lP,m}} n_3
          \]
          which completes the proof.
      \end{itemize}
  \end{enumerate}
\end{proof}
\begin{lemma} \label{lem:swappingpreserves}
  If $s,s' \in P_{\mf A}$ and $s' \swapping s$, then
  \begin{enumerate}
    \item
      $\enabled(s) = \enabled(s')$,
    \item
      $\cp(s) = \cp(s')$, and
    \item
      $\fp(s) = \fp(s')$.
  \end{enumerate}
\end{lemma}
\begin{proof}
  Induction on $\swapping$. The base case is trivial.
  Consider the case where
  $s  = s_1 \snoc \alpha_2 \snoc \alpha_1 \snoc s_2$ and
  $s' = s_1 \snoc \alpha_1 \snoc \alpha_2 \snoc s_2$.
  Let $\alpha_1 = \lmsg l{a_1}{p_1}{p_1'}{d_1}$ and
      $\alpha_2 = \lmsg l{a_2}{p_2}{p_2'}{d_2}$.
  \begin{enumerate}
    \item
      Induction on the length of $s_2$.
      In the base case, we have (by associativity of $\cup$):
      $\enabled(s_1 \snoc \alpha_2 \snoc \alpha_1)
         = \enabled(s_1) \cup \{\tuple{a,p_2'} \mid a_2 \enables A a\}
                         \cup \{\tuple{a,p_1'} \mid a_1 \enables A a\}
         = \enabled(s_1) \cup \{\tuple{a,p_1'} \mid a_1 \enables A a\}
                         \cup \{\tuple{a,p_2'} \mid a_2 \enables A a\}$.
    \item
      Induction on the length of $s_2$ as in 1.
    \item
      Induction on the length of $s_2$.
      In the base case, we have (since by the def. of $\swapping$, $p_1 \neq p_2'$ and $p_2 \neq p_1'$):
      \begin{align*}
           \fp(s_1 \snoc \alpha_2 \snoc \alpha_1) & =
        \\ \fp(s_1 \snoc \alpha_2) \cup (\{p_1\} \setminus \cp(s_1 \snoc \alpha_2)) & =
        \\ \fp(s_1) \cup (\{p_2\} \setminus \cp(s_1)) \cup
          (\{p_1\} \setminus (\cp(s_1) \cup \{p_2'\})) & =
        \\ \fp(s_1) \cup (\{p_2\} \setminus (\cp(s_1) \cup \{p_1'\})) \cup
          (\{p_1\} \setminus \cp(s_1)) & =
        \\ \fp(s_1) \cup (\{p_1\} \setminus \cp(s_1))
          \cup (\{p_2\} \setminus (\cp(s_1) \cup \{p_1'\})) & =
        \\ \fp(s_1 \snoc \alpha_1) \cup (\{p_2\} \setminus \cp(s_1 \snoc \alpha_1)) & =
        \\ \fp(s_1 \snoc \alpha_1 \snoc \alpha_2)
      \end{align*}
  \end{enumerate}
\end{proof}
\begin{lemma} \label{lem:saturationendmessage}
  Let $S \subseteq P_{\mf A}$ be a saturated set of traces. If $s,s' \in S$ are
  traces such that $s' \swapping s$ and $s \snoc \alpha \in S$, then $s' \snoc
  \alpha \in S$.
\end{lemma}
\begin{proof}
  Induction on $\swapping$. The base case is trivial. We show the case
  of a single swapping.
  If $s' \swapping s$, we have
  $s  = s_1 \snoc \alpha_2 \snoc \alpha_1 \snoc s_2$ and
  $s' = s_1 \snoc \alpha_1 \snoc \alpha_2 \snoc s_2$
  for some $s_1,s_2,\alpha_1, \alpha_2$.
  Obviously, $s' \snoc \alpha \swapping s \snoc \alpha$.

  We have to show that if
  $s \snoc \alpha \in P_{\mf A}$, then $s' \snoc \alpha \in P_{\mf A}$.
  We have to show that $s' \snoc \alpha$ fulfils the legality conditions
  imposed by $P_{\mf A}$:
  \begin{itemize}
    \item
      It is easy to see that $s' \snoc \alpha$ has unique pointers and is correctly labelled.
    \item
      $s' \snoc \alpha$ is justified since $\enabled(s) = \enabled(s')$ by
      Lemma~\ref{lem:swappingpreserves}.
    \item
      To see that $s' \snoc \alpha$ strictly scoped, consider the (``worst'') case
      when
      \[\lmsg lap{p'}d \snoc s_3 \snoc \alpha \subseteq s' \snoc \alpha \text{ and } a \in \ans A\]
      (i.e. we pick the segment that goes right up to the end of the trace).
      We consider the different possibilities of the position
      of this answer message:
      \begin{itemize}
        \item
          If $\lmsg lap{p'}d \subseteq s_1$, then let
          $s_4' = \lmsg lap{p'}d \snoc s_1' \snoc \alpha_1 \snoc \alpha_2
                  \snoc s_2 \snoc \alpha \subseteq s' \snoc \alpha$ and
          $s_4 = \lmsg lap{p'}d \snoc s_1' \snoc \alpha_2
                 \snoc \alpha_1 \snoc s_2 \snoc \alpha$.
          We also know that $p \notin \fp(s_4)$ as $s \snoc \alpha \in P_{\mf A}$.
          Now, since $s_4' \swapping s_4$, we have $\fp(s_4) = \fp(s_4')$ by
          Lemma~\ref{lem:swappingpreserves} and thus also $p \notin \fp(s_4')$.
        \item
          If $\lmsg lap{p'}d = \alpha_2$. We know that $p \notin \fp(s_2 \snoc \alpha)$
          by $s \snoc \alpha \in P_{\mf A}$.
          Since $s' \in P_{\mf A}$ we have $p \notin \fp(\alpha_1)$ and can so
          conclude that $p \notin \fp(\alpha_1 \snoc s_2 \snoc \alpha)$.
        \item
          If $\lmsg lap{p'}d = \alpha_1$ or $\lmsg lap{p'}d \subseteq s_2$,
          $p \notin \fp(s_2 \snoc \alpha)$ follows immediately from
          $s \in P_{\mf A}$.
        \item
          If $\lmsg lap{p'}d = \alpha$, $p \notin \fp(\emptytrace) = \emptyset$ is
          trivially true.
      \end{itemize}
    \item
      To see that $s' \snoc \alpha$ is strictly nested, assume
      \[\lmsg {l_1}{a_1}p{p'}{d_1} \snoc s_1 \snoc \lmsg {l_2}{a_2}{p'}{p''}{d_2} \snoc
       s_2  \snoc \lmsg {l_3}{a_3}{p'}{p'''}{d_3}
       \subseteq s' \snoc \alpha \]
       for port names $a_1,a_2 \in \qst A$ and $a_3 \in \ans A$.
       We have to show that this implies
       $\lmsg {l_4}{a_4}{p''}{-}{d_4} \subseteq s_2$, for a port name $a_4 \in \ans A$.
       We proceed by considering the possible positions of the last message in the
       segment:
       \begin{itemize}
         \item
           If $\lmsg {l_3}{a_3}{p'}{p'''}{d_3} \subseteq s'$, then the proof is
           immediate, by $s' \in P_{\mf A}$ being strictly nested.
        \item
          If $\lmsg {l_3}{a_3}{p'}{p'''}{d_3} = \alpha$ we use the fact that
          $s \snoc \alpha \in P_{\mf A}$ is strictly nested. We assume that
          the implication (using the same names) as above holds but instead for
          $s \snoc \alpha$, and show that any
          swappings that can have occurred in $s'$ that reorder the $a_1, a_2, a_4$
          moves would render $s'$ illegal:
          \begin{itemize}
            \item
              If $a_2$ was moved before $a_1$, then $s'$ would not be justified.
            \item
              If $a_4$ was moved before $a_2$, then $s'$ would not be justified.
          \end{itemize}
          As the order is preserved, this shows that the swappings must be done
          in a way such that the implication holds for $s' \snoc \alpha$.
       \end{itemize}
  \end{itemize}
\end{proof}
\begin{lemma} \label{thm:unravelling}
  For any game net $f = \gnet S {\mf A}$ and trace $s \in P_{\mf A}$, $s \in
  \denot f$ if and only if $\forall p \in \fp(s).  s \hereditary \{p\} \in
  \denot f$.
\end{lemma}
\begin{lemma} \label{thm:copycatsubset}
  $\Copycat_{\mf A, \pi\cdot\mf {A}}\cond{st,alt} \subseteq \denot{\copycat_{\pi,\mf A}}$.
\end{lemma}
\begin{proof}
  For convenience,
    let $\gnet f {\mf A \fun \mf{A'}} = \copycat_{\pi,\mf A, \mf{A'}}$,
      $S_1 = \Copycat_{\mf A, \mf {A'}}\cond{st,alt}$ and
      $S_2 = \denot{f}$.
  We show that $s \in S_1$ implies $s \in S_2$, by induction on the length of
  $s$:

  \begin{itemize}
    \item Hypothesis.
      If $s$ has even length, then
      $\initial(f) \netstep s \netconf {\{\engineconf \emptyset h \hastype E\}} \emptyset$
      and $h$ is exactly (nothing more than) a copycat heap for $s$ over $\mf A
      \fun \mf {A'}$.  In other words, there are no threads running and no
      pending messages and the heap is precisely specified.
    \item Base case. Trivial.
    \item Inductive step.
    At any point in the execution of the configuration of $f$, an
    $\lO$-labelled message can be received, so that case is rather
    uninteresting. Since the trace $s$ is alternating, we consider two messages
    in each step:

    Assume $s = s'\snoc \lmsg \lO{a_1}{p_1}{p_1'}{d_1} \snoc
    \lmsg \lP{a_2}{p_2}{p_2'}{d_2} \in S_1$ and that $s' \in S_2$.
    From the definition of $\Copycat$ we know that $a_2 = \bijpi_\portnames(a_1)$,
    $p_2 = \bijpi_\pointers(p_1)$, $p_2' = \bijpi_\pointers(p_2)$,
    and $d_1 = d_2$.

    We are given that $\initial(f) \netstep s' \netconf {\{\engineconf
    \emptyset h \hastype E\}} \emptyset$ as in the hypothesis. We have five
    cases for the port name $a_1$. We show the first three, as the others are
    similar. In each case our single engine will receive a message and start a
    thread:

    \begin{itemize}
      \item If $a_1 \in \ini {A'}$, then (since $s$ is justified) $p_2 = p_1'$
      and (by the definition of $\pi_\portnames'$) $a_2 = \pi_\portnames^{-1}(a_1)$.  The engine runs
      the first clause of the copycat definition, and chooses to create the
      pointer $p_2$ and then performs a send operation.  We thus get:
        \[
          \initial(f) \netstep s \netconf
            {\{\engineconf \emptyset {h \cup \{p_2' \mapsto p_1'\}}\}}
            \emptyset
        \]
        It can easily be verified that the hypothesis holds for this new state.
      \item If $a_1 \in (\opp {A'} \cap \qst {A'}) \setminus \ini {A'}$, then
        $a_2 = \pi_\portnames^{-1}(a_1)$.
        Since $s$ is justified and strictly nested, there is a message $\lmsg
        \lP{a_3}{p_3}{p_1}{d_3} \subseteq s'$ that is pending.

        By the hypothesis
        there is a message $\lmsg \lO {\pi_\portnames'(a_3)}{p_4}{p_4'}{d_4} \subseteq s'$ with
        $h(p_1) = p_4'$, which means that the $\icopyq$ instruction can be run,
        yielding the following:
        \[
          \initial(f) \netstep s \netconf
          {\{\engineconf \emptyset {h \cup \{p_2' \mapsto p_1'\}}\}}
          \emptyset
        \]
        The hypothesis can easily be verified also in this new state.
      \item If $a_1 \in \opp {A'} \cap \ans {A'}$, then $a_2 = \pi_\portnames^{-1}(a_1)$.
        Since $s$ is justified and strictly nested, there is a prefix $s_1
        \snoc \lmsg \lP{a_3}{p_3}{p_1}{d_3} \leq s'$ whose last message is
        a pending question. By the hypothesis $s_1$ is then on the form $s_1 = s_2 \snoc
        \lmsg \lO
        {\pi_\portnames'(a_3)}{\bijpi_\pointers(p_3)}{\bijpi_\pointers(p_1)}{d_4}$
        with $h = h' \cup \{p_1 \mapsto \bijpi_\pointers(p_1)\}$, which means
        that the $\icopya$
        instruction can be run, yielding the following:
        \[
          \initial(f) \netstep s \netconf
          {\{\engineconf \emptyset {h'}\}}
          \emptyset
        \]
        The hypothesis is still true; the $a_3$ question is no longer pending
        and its pointer is removed from the heap (notice that $p_2 =
        \bijpi_\pointers(p_1)$).

    \end{itemize} 
  \end{itemize}
\end{proof}

\begin{theorem} \label{thm:copycatpmoves}
  If $s = s_1 \snoc o \snoc s_2 \in \Copycat_{\mf A, \mf A'}$ and $p \nsubseteq s_2$, then
    $s \snoc p \in \Copycat_{\mf A, \mf A'}$,
  where
  $o = \lmsg \lO ap{p'}{d}$ and
  $p = \pilmsg \lP ap{p'}d$ (i.e. the ``copy'' of $o$).
\end{theorem}
\begin{proof}
  By induction on $\swapping$.
  \begin{itemize}
    \item Base case.
      This means that
      $s = s_1 \snoc o \snoc s_2 \in \Copycat_{\mf A, \mf A'}\cond{alt}$.
      But since $p \nsubseteq s_2$ and by the definition of the alternating copycat,
      $s_2 = \emptytrace$.
      It is easy to check that $s \snoc p \in \Copycat_{\mf A, \mf A'}\cond{alt}$
      and that it is legal.
    \item Inductive step.
      Assume $s \swapping s'$ for an $s' \in P_{\mf {A \fun A'}}$ such that
      $s' \snoc p \in \Copycat_{\mf A, \mf A'}$. By Lemma~\ref{lem:saturationendmessage},
      $s \snoc p \in \Copycat_{\mf A, \mf A'}$.
  \end{itemize}
\end{proof}
\begin{definition}
  Define the multiset of messages that a net configuration $n$ is ready to
  immediately send as $\ready(n) \defeq \{ \tuple {\lP,m} \mid \exists n'. \: n
  \netstep {}^* \netstep {\tuple{\lP,m}} n' \}$.
\end{definition}
\begin{definition} \label{def:copycatheap}
  If $s$ is a trace, $h$ is a heap, $\mf A$ is a game interface,
  and $\pi_\pointers$ is a permutation over $\pointers$,
  we say that $h$ is a \emph{copycat heap} for $s$ over $\mf A$ if and only if:

  For every pending $\lP$-question from $\mf A$ in $s$, i.e.
  $\lmsg \lP ap{p'}d \subseteq s$
  ($a \in \qst A$), $h(p') = \tuple{\bijpi_\pointers(p'),\emptyset}$.
\end{definition}
\begin{lemma} \label{lem:copycatready}
  If $s \in \Copycat$ is a trace such that $\initial(\copycat) \netstep s n$, then
  the following holds:
  \begin{enumerate}
    \item
      If $n \netstep {}^* n'$ then $\ready(n) = \ready(n')$.
    \item
      If $n \netstep {}^* \netstep {\tuple {\lP,m}} n'$, then
      $\ready(n) = \ready(n') \cup \{\tuple {\lP,m}\}$.
  \end{enumerate}
\end{lemma}
As we are only interested in what is observable, the trace $s$ is thus
equivalent to one where silent steps are only taken in one go by one thread
right before outputs.
\begin{proof}
  \begin{enumerate}
    \item
      For convenience, we give the composition of silent steps a name,
      $n \netstep {(x)}^* n'$. We proceed by induction on the length of $(x)$:
      \begin{itemize}
        \item Base case. Immediate.
        \item Inductive step. If $n \netstep {} \netstep {(x')}^* n'$, we analyse
        the first silent step, which means that a thread $t$ of the engine in the
        net takes a step:
        \begin{itemize}
          \item In the cases where an instruction that does not change or depend
            on the heap is run, the step cannot affect $\ready(n)$.
          \item In the case where the instruction is in $\{\icopyi, \icopyq,
            \iexti, \iextq\}$, we note that the heap is not \emph{changed},
            but merely extended with a fresh mapping which can not have appeared
            earlier in the trace.
          \item If the instruction is $\icopya$, since the trace $s$ is strictly
            nested by assumption, the input message that this message stems from
            occurs in a position in the trace where it would later be illegal
            to mention the deallocated pointer again.
        \end{itemize}
      \end{itemize}

    \item
      Immediate.
  \end{enumerate}
\end{proof}
\begin{theorem} \label{thm:copycatplegal}
  If $s \in \Copycat\cond{st}$ is a trace such that $\initial(\copycat) \netstep s n$
  for an $n = \netconf {\{\engineconf {\many t} h \hastype E\}} {\many m}$, then there
  exists a permutation $\pi_\pointers$ over $\pointers$ such that the following holds:
  \begin{enumerate}
    \item
      The heap $h$ is a copycat heap for $s$ over $\mf {A \fun A'}$.
    \item
      The set of messages that $n$ can immediately send, $\ready(n)$, is exactly the
      set of messages $p$ such that
      $s = s_1 \snoc o \snoc s_2$ and $p \nsubseteq s_2$
      where the form of $o$ and $p$ is
      $o = \lmsg \lO ap{p'}{d}$ and
      $p = \pilmsg \lP ap{p'}d$ (i.e. the ``copy'' of $o$).
  \end{enumerate}
\end{theorem}
\begin{proof}
  Induction on the length of $s$.
  The base case is immediate.

  We need to show that if the theorem holds for a trace $s$, then it
  also holds for $s \snoc \alpha$. We thus assume that there exists a permutation
  $\pi_\pointers$ such that the hypothesis holds for $s$ and that
  $\initial(\copycat) \netstep s n \netstep {}^* \netstep \alpha n'$.

  \begin{enumerate}
    \item
      If $\alpha = \pilmsg \lP ap{p'}d$ then by (2) there must be a message
      $o = \lmsg \lO ap{p'}d$ such that $s = s_1 \snoc o \snoc s_2$ and
      $\alpha \in \ready(n)$.
      Since we ``chose'' $\pi_\pointers$ such that $p$ can only be gotten
      from the thread spawned by $o$, we can proceed by cases as we did
      Theorem~\ref{thm:copycatsubset} to see that the heap structure is correct
      in each case.
    \item
      \begin{itemize}
        \item
          If $\alpha = \pilmsg \lP ap{p'}d$ then by (2) there must be a message
          $o = \lmsg \lO ap{p'}d$ such that $s = s_1 \snoc o \snoc s_2$ and
          $\alpha \in \ready(n)$.
          By Lemma~\ref{lem:copycatready}, $\ready(n) = \ready(n') \cup \{\alpha\}$.
          We can easily verify that (2) holds for $n'$.
        \item
          If $\alpha = \lmsg \lO a{p_1}{p_1'}d$, then we can proceed as in
          Theorem~\ref{thm:copycatsubset} to see that a message
          $p = \lmsg \lP {\bijpi_\portnames(a)}{p_2}{p_2'}d \in \ready(n')$.
          We then simply construct our extended permutation such that the
          hypothesis holds.
      \end{itemize}
  \end{enumerate}
\end{proof}

%% file: k.latex
\setlength{\unitlength}{4144sp}%
\begingroup\makeatletter\ifx\SetFigFont\undefined%
\gdef\SetFigFont#1#2#3#4#5{%
  \reset@font\fontsize{#1}{#2pt}%
  \fontfamily{#3}\fontseries{#4}\fontshape{#5}%
  \selectfont}%
\fi\endgroup%
\begin{picture}(2454,2049)(1159,-1423)
\thinlines
{\color[rgb]{0,0,0}\put(1711,-511){\framebox(450,450){}}
}%
{\color[rgb]{0,0,0}\put(1711,-1411){\framebox(450,450){}}
}%
{\color[rgb]{0,0,0}\put(2611,-1411){\framebox(450,1800){}}
}%
{\color[rgb]{0,0,0}\put(3376,-1186){\vector( 1, 0){225}}
}%
{\color[rgb]{0,0,0}\multiput(3061,164)(-60.00000,0.00000){8}{\line(-1, 0){ 30.000}}
\put(2611,164){\vector(-1, 0){0}}
}%
{\color[rgb]{0,0,0}\multiput(2611,-1186)(60.00000,0.00000){8}{\line( 1, 0){ 30.000}}
\put(3061,-1186){\vector( 1, 0){0}}
}%
{\color[rgb]{0,0,0}\multiput(2611,-286)(55.38462,0.00000){7}{\line( 1, 0){ 27.692}}
\multiput(2971,-286)(0.00000,-60.00000){8}{\line( 0,-1){ 30.000}}
\multiput(2971,-736)(-55.38462,0.00000){7}{\line(-1, 0){ 27.692}}
\put(2611,-736){\vector(-1, 0){0}}
}%
{\color[rgb]{0,0,0}\put(2296,-736){\line(-1, 0){810}}
\put(1486,-736){\line( 0,-1){450}}
\put(1486,-1186){\vector( 1, 0){225}}
}%
{\color[rgb]{0,0,0}\put(2296,164){\line(-1, 0){810}}
\put(1486,164){\line( 0,-1){450}}
\put(1486,-286){\vector( 1, 0){225}}
}%
{\color[rgb]{0,0,0}\put(2161,-286){\line( 1, 0){135}}
}%
{\color[rgb]{0,0,0}\put(2476,-286){\vector( 1, 0){135}}
}%
{\color[rgb]{0,0,0}\put(2161,-1186){\line( 1, 0){135}}
}%
{\color[rgb]{0,0,0}\put(2476,-1186){\vector( 1, 0){135}}
}%
{\color[rgb]{0,0,0}\put(2611,-736){\line(-1, 0){135}}
}%
{\color[rgb]{0,0,0}\put(3061,-1186){\line( 1, 0){135}}
}%
{\color[rgb]{0,0,0}\put(3196,164){\vector(-1, 0){135}}
}%
{\color[rgb]{0,0,0}\put(1171,614){\line( 1, 0){2295}}
\put(3466,614){\line( 0,-1){450}}
\put(3466,164){\line(-1, 0){ 90}}
}%
{\color[rgb]{0,0,0}\put(2611,164){\line(-1, 0){135}}
}%
\put(1936,-1231){\makebox(0,0)[b]{\smash{{\SetFigFont{10}{12.0}{\familydefault}{\mddefault}{\updefault}{\color[rgb]{0,0,0}$g$}%
}}}}
\put(1936,-331){\makebox(0,0)[b]{\smash{{\SetFigFont{10}{12.0}{\familydefault}{\mddefault}{\updefault}{\color[rgb]{0,0,0}$f$}%
}}}}
\put(2836,-556){\makebox(0,0)[b]{\smash{{\SetFigFont{10}{12.0}{\familydefault}{\mddefault}{\updefault}{\color[rgb]{0,0,0}$K$}%
}}}}
\put(2386,-331){\makebox(0,0)[b]{\smash{{\SetFigFont{10}{12.0}{\familydefault}{\mddefault}{\updefault}{\color[rgb]{0,0,0}$B$}%
}}}}
\put(2386,119){\makebox(0,0)[b]{\smash{{\SetFigFont{10}{12.0}{\familydefault}{\mddefault}{\updefault}{\color[rgb]{0,0,0}$A$}%
}}}}
\put(2386,-781){\makebox(0,0)[b]{\smash{{\SetFigFont{10}{12.0}{\familydefault}{\mddefault}{\updefault}{\color[rgb]{0,0,0}$B$}%
}}}}
\put(2386,-1231){\makebox(0,0)[b]{\smash{{\SetFigFont{10}{12.0}{\familydefault}{\mddefault}{\updefault}{\color[rgb]{0,0,0}$C$}%
}}}}
\put(3286,-1231){\makebox(0,0)[b]{\smash{{\SetFigFont{10}{12.0}{\familydefault}{\mddefault}{\updefault}{\color[rgb]{0,0,0}$C$}%
}}}}
\put(3286,119){\makebox(0,0)[b]{\smash{{\SetFigFont{10}{12.0}{\familydefault}{\mddefault}{\updefault}{\color[rgb]{0,0,0}$A$}%
}}}}
\put(1396,-106){\makebox(0,0)[b]{\smash{{\SetFigFont{10}{12.0}{\familydefault}{\mddefault}{\updefault}{\color[rgb]{0,0,0}$\eta$}%
}}}}
\put(1396,-1006){\makebox(0,0)[b]{\smash{{\SetFigFont{10}{12.0}{\familydefault}{\mddefault}{\updefault}{\color[rgb]{0,0,0}$\eta$}%
}}}}
\put(3556,344){\makebox(0,0)[b]{\smash{{\SetFigFont{10}{12.0}{\familydefault}{\mddefault}{\updefault}{\color[rgb]{0,0,0}$\varepsilon$}%
}}}}
\end{picture}%

%% file: k.tex
\begin{definition} \label{def:extcopycatheap} If $s$ is a trace, $h$ is a heap, $\mf A$ is a game interface,
  and $\pi_\pointers$ is a permutation over $\pointers$,
  we say that $h$ is an \emph{extended copycat heap} for $s$ over $\mf A$ if and only if:

  \begin{enumerate}
    \item
      For every pending $\lP$-question non-initial in $\mf A$ in $s$, i.e.
      $\lmsg \lP ap{p'}d \subseteq s$ ($a \in \qst A \setminus \ini A$), $h(p') =
      \tuple{\bijpi_\pointers(p'),\emptyset}$.
    \item
      For every pending $\lP$-question initial in $\mf A$ in $s$ and its justifying move, i.e.
      $\lmsg \lO {a_1}{p_1}{p}{d_1} \snoc s' \snoc \lmsg \lP {a_2}{p}{p_2}{d_2} \subseteq s$
      ($a_2 \in \ini A$),
      $h(p_2) = \tuple{\bijpi_\pointers(p_2),\bijpi_\pointers(p_1)}$.
  \end{enumerate}
\end{definition}
\begin{theorem} \label{thm:Ksubset}
  If $f \hastype \mf {A \tfun B}$ and $g \hastype \mf {B' \tfun C}$ are game
  nets such that $\pi_{\mf B} \vdash \mf B \pneq \mf B'$,
  $f$ implements $S_f \subseteq P_{\mf {A \fun B}}$, and
  $g$ implements $S_g \subseteq P_{\mf {B' \fun C}}$, then
  $(S_f \gcompose S_g)\cond {st,alt} \subseteq_{\portnames\pointers} \denot {f \compose_{GAM} g} =
  \denot{\Lambda^{-1}_{A}(\Lambda_A(f)\tensor\Lambda_{B'}(g)\compose \gcomp_{\mf A,\mf B,\mf C})}$.
\end{theorem}
\begin{proof}
  We show that $s' \in (S_f \gcompose S_g)\cond{st,alt}$ implies that there
  exists a $\pi_\pointers$ such that
  $\pi_{\mf {A,C}} \cdot \pi_\pointers \cdot s' \in \denot {f \compose_{GAM} g}
   = \denot{\Lambda^{-1}_{A}(\Lambda_A(f)\tensor\Lambda_{B'}(g)\compose \gcomp_{\mf A,\mf B,\mf C})}
   = \denot{\Lambda_A(f)} \tensor \denot{\Lambda_{B'}(g)} \compose \denot{\gcomp_{\mf A,\mf B,\mf C}}$.
  Recall the definition of game composition:
  \[ S_f \gcompose S_g \defeq
      \{s \del B \mid s \in \traces_{A \tensor B \tensor C}
        \wedge s \del C \in S_f
        \wedge \pi_{\mf B} \cdot s^{*B} \del A \in S_g \}
  \]
  We proceed by induction on the length of such an $s$:
  \begin{itemize}
    \item Hypothesis.
      There exists an $s_\gcomp$ such that
      $\initial(\gcomp_{\mf {A,B,C}}) \netstep {s_\gcomp} n$ where $n = \netconf {\{\engineconf \emptyset h \hastype E\}} \emptyset$
      and $h$ is exactly (nothing more than) the union of a copycat heap for
      $s_\gcomp$ over $\mf {A' \fun A}$, a copycat heap for $s_\gcomp$ over $\mf {C \fun C'}$ and
      an extended copycat heap for $s_\gcomp$ over $\mf {B \fun B'}$.

      Let
      \begin{align*}
           s_f          & \defeq s \del C
        \\ s_g          & \defeq \pi_{\mf B} \cdot s^{*B} \del A
        \\ s_{f \compose g} & \defeq s \del B
        \\ s_{\gcomp f} & \defeq s_\gcomp \ndel A',B',C,C'
          \text{, the part of $s_\gcomp$ relating to $f$}
        \\ s_{\gcomp g} & \defeq s_\gcomp \ndel A,A',B,C'
          \text{, the part of $s_\gcomp$ relating to $g$}
        \\ s_{\gcomp f \compose g} & \defeq s_\gcomp \ndel A,B,B',C
          \text{, the part of $s_\gcomp$ relating to the whole game net.}
      \end{align*}
      We require that $s_\gcomp$ fulfils $s_{\gcomp f}\dual = s_f$, $s_{\gcomp g}\dual = s_g$,
      and $s_{\gcomp f \compose g} = \pi_{\mf {A,C}} \cdot \pi_\pointers \cdot s_{f \compose g}$.
      Note that $s_{\gcomp f\compose g}$ is the trace of $f \compose_{GAM} g$, by
      the definition of trace composition.

    \item Base case.
      Immediate.
    \item Inductive step.
      Assume $s = s' \snoc \alpha$ and that the hypothesis holds for $s'$ and
      some $\pi_\pointers'$ and $s'_\gcomp$. We
      proceed by cases on the $\alpha$ message:
      \begin{itemize}
        \item
          If $\alpha = \lmsg \lO ap{p'}d$, we have three cases:
          \begin{itemize}
            \item If $a \in \supp(A)$,
                  intuitively this means that we are getting a message from outside
                  the $\gcomp$ engine, and need to propagate it through $\gcomp$ to $f$.
                  We construct $s_\gcomp$ and $\pi_\pointers$, such that
                  $s_\gcomp = s'_\gcomp \snoc \lmsg \lO {\pi_{\mf
                  A}(a)}{\bijpi_\pointers(p)}{\bijpi_\pointers(p')}d \snoc
                  \alpha\dual$, by further sub-cases on $a$ ($\pi_\pointers$ will be
                  determined by steps of the $\gcomp$ configuration):

              \begin{itemize}
                \item
                  $a \in \ini A$ cannot be the case because an initial message in $A$
                  must be justified by an initial ($\lO$-message) in $C$, and so must be a $\lP$-message.
                \item
                  If $a \in (\qst A \setminus \ini A) \cup \ans A$, this means that $s' \del
                  C \snoc \alpha = (s' \snoc \alpha) \del C$ as the message
                  must be justified by a message from $\mf A$. As $f$ is $\lO$-closed
                  $s \del C \in \denot {\Lambda_A(f)}$.
                  This trace can be stepped to by $n'$ just like how it was done
                  in Theorem~\ref{thm:copycatsubset}. We can verify that
                  the parts of the hypothesis not in that theorem hold -- in particular
                  for this case we have
                  $s_{\gcomp f} = s'_{\gcomp f} \snoc \alpha\dual$, so indeed
                  $s_{\gcomp f}\dual = s_f$ as required.
              \end{itemize}
            \item $a \in \supp(B)$:

              Intuitively this means that $g$ is sending a message to $f$, which
              has to go through $\gcomp$.
              We construct $s_\gcomp$ and $\pi_\pointers$, such that
              $s_\gcomp = s'_\gcomp \snoc \lmsg \lO {a}{\bijpi_\pointers(p)}{\bijpi_\pointers(p')}d \snoc
              \pi_{\mf B}\cdot \alpha\dual$, by further sub-cases on $a$ ($\pi_\pointers$ will be
              determined by steps of the $\gcomp$ configuration):
              \begin{itemize}
                \item
                  If $a \in \ini B$, there must be a pending $\lP$-message from $\mf C$
                  justifying $\alpha$ in $s'$, i.e. $\lmsg \lP {a_0}{p_0}{\bijpi_\pointers(p)}{d_0} \subseteq s'$
                  and then by Definition~\ref{def:copycatheap}
                  $h(\bijpi_\pointers(p)) = \tuple{p,\emptyset}$ (as $\bijpi_\pointers$ is its own inverse).
                  This means that (running the $\iexti$ instruction) we get:
                  \begin{align*}
                  n' \netstep
                          {\lmsg \lO {\pi_{\mf B}(a)}{\bijpi_\pointers(p)}{\bijpi_\pointers(p')}d}
                      \netstep{}^*
                      \netstep{\alpha\dual} & \\
                    \netconf {\{\engineconf \emptyset {h \cup \{ p' \hmap {\bijpi_\pointers(p'),p}\}} \hastype E\}} \emptyset = & n
                  \end{align*}
                  Now $\pi_{\mf B} \cdot \alpha\dual$ is a new pending
                  $\lP$-question in the trace that is initial in $\mf{B \fun
                  B'}$, but our new heap mapping fulfils clause (2) of
                  Definition~\ref{def:extcopycatheap} as required.
                \item
                  If $a \in (\qst B \setminus \ini B) \cup \ans B$, this is
                  similar to the $\mf A$ case (note that the extended copycat only
                  differs from the ordinary copycat for initial messages).
              \end{itemize}
            \item If $a \in \supp(C)$.

              Intuitively this means that we are getting a message from outside
              the $\gcomp$ engine, and need to propagate it through $\gcomp$ to $g$.
              We construct $s_\gcomp$ and $\pi_\pointers$, such that:
              \[s_\gcomp = s'_\gcomp \snoc \lmsg \lO {\pi_{\mf
              C}(a)}{\bijpi_\pointers(p)}{\bijpi_\pointers(p')}d \snoc
              \alpha\dual\]
              In this case, the code that we will run is just that of
              $\copycat$, so we can proceed like in
              Theorem~\ref{thm:copycatsubset}, easily verifying our additional
              assumptions.
          \end{itemize}
        \item
          If $\alpha = \lmsg \lP ap{p'}d$, we have three cases:
          \begin{itemize}
            \item If $a \in \supp(A)$, intuitively this means that we get
              a message from $f$ and need to propagate it through $\gcomp$ to
              the outside.
              By further sub-cases on $a$, we construct $s_\gcomp$ and $\pi_\pointers$, such that:
              \[s_\gcomp = s'_\gcomp \snoc \alpha\dual \snoc \lmsg \lP {\pi_{\mf A}(a)}{\bijpi_\pointers(p)}{\bijpi_\pointers(p')}d\]
              The pointer permutation $\pi_\pointers$ will be determined by steps of the $\gcomp$ configuration.
              \begin{itemize}
                \item
                  If $a \in \ini A$,
                  then $\alpha$ must be justified in $s'$ by a pending
                  and initial $\lP$-question from $\mf B$ by the definition of
                  $\mf A \fun \mf B$ which must in turn be justified by a pending
                  and initial $\lO$-question from $\mf C$ by the definition of
                  $\mf B \fun \mf C$. In $s'_\gcomp$, we have
                  (since $s'_{\gcomp f\compose g} = \pi_{\mf {A,C}} \cdot \pi_\pointers \cdot s'_{f \compose g}$)
                    \[s'_\gcomp =
                      s_1 \snoc \lmsg \lO {a_{\mf C'}}{p_0}{p_{\mf C'}}{d_{\mf C'}} \snoc
                      s_2 \snoc \lmsg \lP {a_{\mf B}}{p_{\mf C'}}{p}{d_{\mf C'}} \snoc
                      s_3\]

                  This means that clause (2) in Definition~\ref{def:extcopycatheap}
                  applies, such that $h(p) = \tuple{\bijpi_\pointers(p),\bijpi_\pointers(p_0)}$
                  and that (running the $\iextq$ instruction) we get:
                   \begin{align*}
                     n'
                        \netstep{\alpha\dual}
                        \netstep{}^* \netstep
                        {\lmsg \lP {\pi_{\mf A}(a)}{\bijpi_\pointers(p)}{\bijpi_\pointers(p')}d} & \\
                      \netconf {\{\engineconf \emptyset {h \cup \{ \bijpi_\pointers(p') \hmap {p',d}\}} \hastype E\}} \emptyset = & n
                    \end{align*}

                  Clause (1) of Definition~\ref{def:extcopycatheap} applies to these
                  new messages and trivially holds.

                \item
                  When $a \in (\qst A \setminus \ini A) \cup \ans A$,
                  the code that we will run is just that of
                  $\copycat$, so we can proceed like in
                  Theorem~\ref{thm:copycatsubset}, also verifying our additional
                  assumptions.
              \end{itemize}
            \item If $a \in \supp(B)$, intuitively this means that $f$ is
              sending a message to $g$, which has to go through $\gcomp$.
              \begin{itemize}
                \item
                  $a \in \ini B$ cannot be the case for a $\lP$-message.
                \item
                  When $a \in (\qst B \setminus \ini B) \cup \ans B$,
                  the code that we will run is just that of
                  $\copycat$, so we can proceed like in
                  Theorem~\ref{thm:copycatsubset}, also verifying our additional
                  assumptions.
              \end{itemize}
            \item If $a \in \supp(C)$, intuitively this means that we get
              a message from $g$ and need to propagate it through $\gcomp$ to
              the outside.
              \begin{itemize}
                \item
                  $a \in \ini C$ cannot be the case for a $\lP$-message.
                \item
                  When $a \in (\qst C \setminus \ini C) \cup \ans C$,
                  the code that we will run is just that of
                  $\copycat$, so we can proceed like in
                  Theorem~\ref{thm:copycatsubset}, also verifying our additional
                  assumptions.
              \end{itemize}
          \end{itemize}
      \end{itemize}
  \end{itemize}
\end{proof}
\begin{lemma} \label{thm:Kplegal}
  If $f \hastype \mf {A \tfun B}$ and $g \hastype \mf {B' \tfun C}$ are game
  nets such that $\pi_{\mf B} \vdash \mf B \pneq \mf B'$,
  $f$ implements $S_f \subseteq P_{\mf {A \fun B}}$, and
  $g$ implements $S_g \subseteq P_{\mf {B' \fun C}}$, then
  $\denot {(f \compose_{GAM} g)}$ is $\lP$-closed with respect to
  $(S_f \gcompose S_g)$.
\end{lemma}
\begin{proof}
  Similar to Theorems~\ref{thm:copycatplegal} and~\ref{thm:Ksubset}.
  We identify the set $\ready(n)$ with ``uncopied'' messages of a $\gcomp$
  net configuration $n$ and show
  that these are legal according to the game composition. Then we show by
  induction that, assuming a heap as in Theorem~\ref{thm:Ksubset}, the
  $\ready(n)$ set is precisely those messages.
\end{proof}

%% file: app.latex
\setlength{\unitlength}{4144sp}%
\begingroup\makeatletter\ifx\SetFigFont\undefined%
\gdef\SetFigFont#1#2#3#4#5{%
  \reset@font\fontsize{#1}{#2pt}%
  \fontfamily{#3}\fontseries{#4}\fontshape{#5}%
  \selectfont}%
\fi\endgroup%
\begin{picture}(3939,2499)(-776,-1783)
\thinlines
{\color[rgb]{0,0,0}\put(811,-1096){\line( 0,-1){225}}
}%
{\color[rgb]{0,0,0}\put(901,-1096){\line( 0,-1){225}}
}%
{\color[rgb]{0,0,0}\put(2701,-1096){\line( 0,-1){225}}
}%
{\color[rgb]{0,0,0}\put(-134,-871){\line(-1, 0){180}}
\put(-314,-871){\line( 0,-1){450}}
}%
{\color[rgb]{0,0,0}\put(316,-961){\line( 1, 0){ 90}}
\put(406,-961){\line( 0,-1){360}}
}%
{\color[rgb]{0,0,0}\put(316,-781){\line( 1, 0){180}}
\put(496,-781){\line( 0,-1){540}}
}%
{\color[rgb]{0,0,0}\multiput(-314,-1321)(0.00000,-64.28571){4}{\line( 0,-1){ 32.143}}
}%
{\color[rgb]{0,0,0}\multiput(2701,-1321)(0.00000,-64.28571){4}{\line( 0,-1){ 32.143}}
}%
{\color[rgb]{0,0,0}\put(496,-1321){\line( 0,-1){ 45}}
\multiput(496,-1366)(57.27273,0.00000){6}{\line( 1, 0){ 28.636}}
\put(811,-1366){\line( 0, 1){ 45}}
}%
{\color[rgb]{0,0,0}\multiput(406,-1321)(0.00000,-60.00000){2}{\line( 0,-1){ 30.000}}
\multiput(406,-1411)(58.23529,0.00000){9}{\line( 1, 0){ 29.118}}
\multiput(901,-1411)(0.00000,60.00000){2}{\line( 0, 1){ 30.000}}
}%
{\color[rgb]{0,0,0}\put(2521,119){\line( 1, 0){180}}
\put(2701,119){\line( 0,-1){765}}
}%
{\color[rgb]{0,0,0}\put(2071,479){\line(-1, 0){270}}
\put(1801,479){\line( 0,-1){1125}}
}%
{\color[rgb]{0,0,0}\put(1441,569){\line( 1, 0){270}}
\put(1711,569){\line( 0,-1){1215}}
}%
{\color[rgb]{0,0,0}\put(1441,389){\line( 1, 0){180}}
\put(1621,389){\line( 0,-1){1035}}
}%
{\color[rgb]{0,0,0}\put(1441,-196){\line( 1, 0){ 90}}
\put(1531,-196){\line( 0,-1){450}}
}%
{\color[rgb]{0,0,0}\put(991,569){\line(-1, 0){180}}
\put(811,569){\line( 0,-1){1215}}
}%
{\color[rgb]{0,0,0}\put(991,-196){\line(-1, 0){ 90}}
\put(901,-196){\line( 0,-1){450}}
}%
{\color[rgb]{0,0,0}\put(2071,-196){\line(-1, 0){ 90}}
\put(1981,-196){\line( 0,-1){450}}
}%
{\color[rgb]{0,0,0}\put(2071,119){\line(-1, 0){180}}
\put(1891,119){\line( 0,-1){765}}
}%
{\color[rgb]{0,0,0}\put(586,-1096){\framebox(2340,450){}}
}%
{\color[rgb]{0,0,0}\put(1711,-646){\line( 0,-1){ 45}}
\multiput(1711,-691)(60.00000,0.00000){2}{\line( 1, 0){ 30.000}}
\put(1801,-691){\line( 0, 1){ 45}}
}%
{\color[rgb]{0,0,0}\multiput(1621,-646)(0.00000,-60.00000){2}{\line( 0,-1){ 30.000}}
\multiput(1621,-736)(60.00000,0.00000){5}{\line( 1, 0){ 30.000}}
\multiput(1891,-736)(0.00000,60.00000){2}{\line( 0, 1){ 30.000}}
}%
{\color[rgb]{0,0,0}\multiput(1531,-646)(0.00000,-55.38462){7}{\line( 0,-1){ 27.692}}
\multiput(1531,-1006)(60.00000,0.00000){8}{\line( 1, 0){ 30.000}}
\multiput(1981,-1006)(0.00000,55.38462){7}{\line( 0, 1){ 27.692}}
}%
{\color[rgb]{0,0,0}\multiput(901,-646)(0.00000,-60.00000){8}{\line( 0,-1){ 30.000}}
}%
{\color[rgb]{0,0,0}\multiput(811,-646)(0.00000,-60.00000){8}{\line( 0,-1){ 30.000}}
}%
{\color[rgb]{0,0,0}\multiput(2701,-646)(0.00000,-60.00000){8}{\line( 0,-1){ 30.000}}
}%
{\color[rgb]{0,0,0}\put(-134,-1096){\framebox(450,450){}}
}%
{\color[rgb]{0,0,0}\put(991,-421){\framebox(450,450){}}
}%
{\color[rgb]{0,0,0}\put(991,254){\framebox(450,450){}}
}%
{\color[rgb]{0,0,0}\put(2701,-1546){\line( 0,-1){225}}
\put(2701,-1771){\line( 1, 0){450}}
}%
{\color[rgb]{0,0,0}\put(-314,-1546){\line( 0,-1){225}}
\put(-314,-1771){\line(-1, 0){450}}
}%
{\color[rgb]{0,0,0}\put(-539,-1546){\framebox(3465,225){}}
}%
{\color[rgb]{0,0,0}\put(2071,-421){\framebox(450,1125){}}
}%
\put(1756,-916){\makebox(0,0)[b]{\smash{{\SetFigFont{10}{12.0}{\familydefault}{\mddefault}{\updefault}{\color[rgb]{0,0,0}$K$}%
}}}}
\put( 91,-916){\makebox(0,0)[b]{\smash{{\SetFigFont{10}{12.0}{\familydefault}{\mddefault}{\updefault}{\color[rgb]{0,0,0}$\delta$}%
}}}}
\put(1216,-241){\makebox(0,0)[b]{\smash{{\SetFigFont{10}{12.0}{\familydefault}{\mddefault}{\updefault}{\color[rgb]{0,0,0}$M'$}%
}}}}
\put(2296,119){\makebox(0,0)[b]{\smash{{\SetFigFont{10}{12.0}{\familydefault}{\mddefault}{\updefault}{\color[rgb]{0,0,0}$\geval$}%
}}}}
\put(1216,434){\makebox(0,0)[b]{\smash{{\SetFigFont{10}{12.0}{\familydefault}{\mddefault}{\updefault}{\color[rgb]{0,0,0}$M$}%
}}}}
\put(-674,-1726){\makebox(0,0)[b]{\smash{{\SetFigFont{10}{12.0}{\familydefault}{\mddefault}{\updefault}{\color[rgb]{0,0,0}$\Gamma$}%
}}}}
\put(1261,-1487){\makebox(0,0)[b]{\smash{{\SetFigFont{10}{12.0}{\familydefault}{\mddefault}{\updefault}{\color[rgb]{0,0,0}$K'$}%
}}}}
\put(3016,-1726){\makebox(0,0)[b]{\smash{{\SetFigFont{10}{12.0}{\familydefault}{\mddefault}{\updefault}{\color[rgb]{0,0,0}$B$}%
}}}}
\end{picture}%

%% file: appopt.latex
\setlength{\unitlength}{4144sp}%
\begingroup\makeatletter\ifx\SetFigFont\undefined%
\gdef\SetFigFont#1#2#3#4#5{%
  \reset@font\fontsize{#1}{#2pt}%
  \fontfamily{#3}\fontseries{#4}\fontshape{#5}%
  \selectfont}%
\fi\endgroup%
\begin{picture}(1824,2049)(439,-1513)
\thinlines
{\color[rgb]{0,0,0}\put(1081,-601){\framebox(450,450){}}
}%
{\color[rgb]{0,0,0}\put(1531,389){\line( 1, 0){270}}
\put(1801,389){\line( 0,-1){1215}}
}%
{\color[rgb]{0,0,0}\put(1531,209){\line( 1, 0){180}}
\put(1711,209){\line( 0,-1){1035}}
}%
{\color[rgb]{0,0,0}\put(1081, 74){\framebox(450,450){}}
}%
{\color[rgb]{0,0,0}\put(1531,-376){\line( 1, 0){ 90}}
\put(1621,-376){\line( 0,-1){450}}
}%
{\color[rgb]{0,0,0}\put(1081,299){\line(-1, 0){180}}
\put(901,299){\line( 0,-1){1125}}
}%
{\color[rgb]{0,0,0}\put(1081,-376){\line(-1, 0){ 90}}
\put(991,-376){\line( 0,-1){450}}
}%
{\color[rgb]{0,0,0}\multiput(901,-826)(0.00000,-60.00000){8}{\line( 0,-1){ 30.000}}
}%
{\color[rgb]{0,0,0}\multiput(1621,-826)(0.00000,-60.00000){2}{\line( 0,-1){ 30.000}}
\multiput(1621,-916)(60.00000,0.00000){2}{\line( 1, 0){ 30.000}}
\multiput(1711,-916)(0.00000,60.00000){2}{\line( 0, 1){ 30.000}}
}%
{\color[rgb]{0,0,0}\multiput(1801,-826)(0.00000,-60.00000){8}{\line( 0,-1){ 30.000}}
}%
{\color[rgb]{0,0,0}\put(901,-1276){\line( 0,-1){225}}
\put(901,-1501){\line(-1, 0){450}}
}%
{\color[rgb]{0,0,0}\put(1801,-1276){\line( 0,-1){225}}
\put(1801,-1501){\line( 1, 0){450}}
}%
{\color[rgb]{0,0,0}\put(676,-1276){\framebox(1350,450){}}
}%
{\color[rgb]{0,0,0}\multiput(991,-826)(0.00000,-64.28571){4}{\line( 0,-1){ 32.143}}
\multiput(991,-1051)(-60.00000,0.00000){2}{\line(-1, 0){ 30.000}}
}%
\put(1306,-421){\makebox(0,0)[b]{\smash{{\SetFigFont{10}{12.0}{\familydefault}{\mddefault}{\updefault}{\color[rgb]{0,0,0}$M'$}%
}}}}
\put(1306,254){\makebox(0,0)[b]{\smash{{\SetFigFont{10}{12.0}{\familydefault}{\mddefault}{\updefault}{\color[rgb]{0,0,0}$M$}%
}}}}
\put(586,-1456){\makebox(0,0)[b]{\smash{{\SetFigFont{10}{12.0}{\familydefault}{\mddefault}{\updefault}{\color[rgb]{0,0,0}$\Gamma$}%
}}}}
\put(2116,-1456){\makebox(0,0)[b]{\smash{{\SetFigFont{10}{12.0}{\familydefault}{\mddefault}{\updefault}{\color[rgb]{0,0,0}$B$}%
}}}}
\put(1351,-1096){\makebox(0,0)[b]{\smash{{\SetFigFont{10}{12.0}{\familydefault}{\mddefault}{\updefault}{\color[rgb]{0,0,0}$@$}%
}}}}
\end{picture}%